\documentclass{article}

\usepackage{microtype}
\usepackage{graphicx}
\usepackage{subcaption}
\usepackage{xurl} % 让 URL 更容易自动换行

\usepackage{booktabs} % for professional tables
\usepackage{bm}

\usepackage{hyperref}
\usepackage{makecell}
\usepackage{multirow} 
\usepackage{url}
\usepackage{breakurl}

\usepackage[preprint]{icml2026}

\usepackage{amsmath}
\usepackage{amssymb}
\usepackage{mathtools}
\usepackage{amsthm}

\usepackage{tabularx}
\usepackage{hhline}
\newcolumntype{Y}{>{\centering\arraybackslash}X}
\newcolumntype{L}[1]{>{\raggedright\arraybackslash}p{#1}}
\usepackage[capitalize,noabbrev]{cleveref}

\theoremstyle{plain}
\newtheorem{theorem}{Theorem}[section]

\newtheorem{lemma}[theorem]{Lemma}

\theoremstyle{definition}
\newtheorem{definition}[theorem]{Definition}

\theoremstyle{remark}
\newtheorem{remark}[theorem]{Remark}

\usepackage[textsize=tiny]{todonotes}

\icmltitlerunning{Differentially Private and Communication Efficient LLM Split Inference}

\begin{document}

\twocolumn[
\icmltitle{Differentially Private and Communication Efficient Large Language Model Split Inference via Stochastic Quantization and Soft Prompt}



\icmlsetsymbol{equal}{*}

\begin{icmlauthorlist}
\icmlauthor{Yujie Gu}{yyy}
\icmlauthor{Richeng Jin}{yyy}
\icmlauthor{Xiaoyu Ji}{yyy}
\icmlauthor{Yier Jin}{xxx}
\icmlauthor{Wenyuan Xu}{yyy}

\end{icmlauthorlist}

\icmlaffiliation{yyy}{Zhejiang University, Hangzhou, China}
\icmlaffiliation{xxx}{University of Science and Technology of China, Anhui, China}

\icmlcorrespondingauthor{Richeng Jin}{richeng\_jin@zju.edu.cn}

\icmlkeywords{Machine Learning, ICML}

\vskip 0.3in
]


\printAffiliationsAndNotice{}  % leave blank if no need to mention equal contribution

\begin{abstract}
Large Language Models (LLMs) have achieved remarkable performance and received significant research interest. The enormous computational demands, however, hinder the local deployment on devices with limited resources. The current prevalent LLM inference paradigms require users to send queries to the service providers for processing, which raises critical privacy concerns. Existing approaches propose to allow the users to obfuscate the token embeddings before transmission and utilize local models for denoising. Nonetheless, transmitting the token embeddings and deploying local models may result in excessive communication and computation overhead, preventing practical implementation. In this work, we propose \textbf{DEL}, a framework for \textbf{D}ifferentially private and communication \textbf{E}fficient \textbf{L}LM  split inference. More specifically, an embedding projection module and a differentially private stochastic quantization mechanism are proposed to reduce the communication overhead in a privacy-preserving manner. To eliminate the need for local models, we adapt soft prompt at the server side to compensate for the utility degradation caused by privacy. To the best of our knowledge, this is the first work that utilizes soft prompt to improve the trade-off between privacy and utility in LLM inference, and extensive experiments on text generation and natural language understanding benchmarks demonstrate the effectiveness of the proposed method.

\end{abstract}

\section{Introduction}
\label{introduction}

Large Language Models (LLMs) have demonstrated remarkable capabilities in natural language processing (NLP), achieving state-of-the-art performance in tasks such as text understanding, reasoning, and generation \cite{brown2020language,zhang2022opt,touvron2023}. Due to their massive scale and resource demands, LLMs are typically deployed as online inference services \cite{LLMSurvey}. However, this service-oriented paradigm raises significant privacy concerns, as user queries often contain sensitive personal or organizational information, including names, addresses, medical records, or trade secrets, which can be exposed to untrusted service providers or intercepted during transmission, as illustrated in Figure \ref{fig:pri_lek}.

\begin{figure}[t]
    \centering
    \includegraphics[width=1\linewidth]{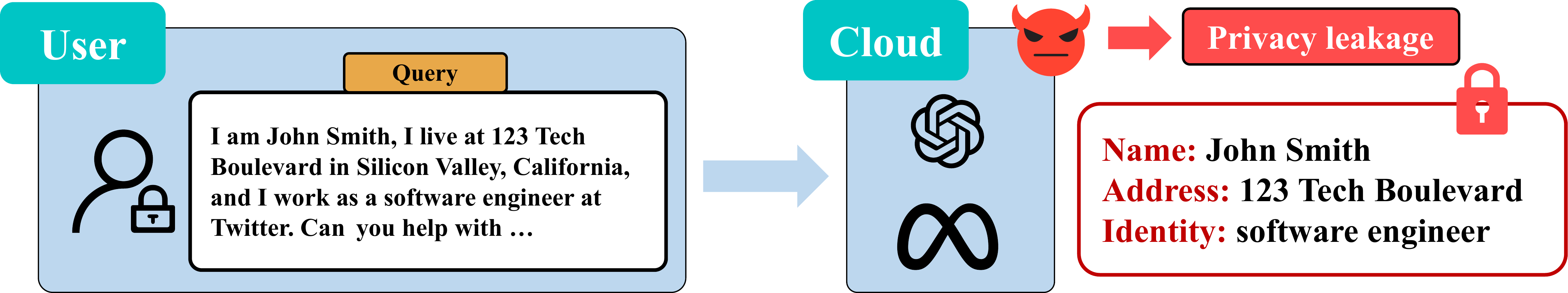}
    \caption{Illustration of privacy leakage during the LLM inference.}
    \label{fig:pri_lek}
    \vspace{-0.2in}
\end{figure}

To address these risks, a variety of privacy-preserving inference methods have been explored. Some works leverage cryptographic techniques such as homomorphic encryption (HE) \cite{liu2023llms,hou2023ciphergpt}. Although HE ensures strong privacy by allowing computation over encrypted inputs, the prohibitive computational overhead renders it impractical for modern transformer-based LLMs. Another prominent line of research is based on perturbation mechanisms, which introduce carefully calibrated noise to user inputs or intermediate representations to achieve differential privacy (DP) guarantees \cite{du2023dpforward,mai2024split}. For example, Split-and-Denoise (SnD) \cite{mai2024split} executes only the token embedding layer of the LLM on the user side and delegates all subsequent layers to the LLM server, thereby adhering to the split inference paradigm \cite{He2019MIACollabInference,liu2024mitigating}, which preserves the privacy of raw text by transmitting only intermediate representations. 
% However, recent research on embedding inversion attacks (EIA) \cite{chen2024text,wan2024information} has revealed that 
Considering that sensitive information may be recovered from the embeddings \cite{chen2024text,wan2024information}, SnD further applies calibrated noise to the token embeddings to ensure local DP. To restore the utility, a pretrained denoising model is employed on the user side to denoise the noisy output from the server. Nonetheless, SnD is designed for relatively simple semantic classification tasks and is not suitable for text generation scenarios. Moreover, the high dimensionality of perturbed token embeddings leads to substantial communication overhead. For instance, a sequence of 1024 tokens in Llama3-8B can incur over 16~MB of communication overhead. To enable private open-ended text generation, InferDPT \cite{tong2025inferdpt} perturbs each token in the input under local DP and employs a local lightweight LLM to extract coherent generations from the perturbed outputs of the LLM server. However, in practice, the local LLM (e.g., Llama2-7B-4bit \cite{touvron2023}) incurs substantial computational and memory costs on the user side. Moreover, the overall performance largely depends on the capability of the local models, which weakens the role of the powerful LLM service provider.

In this work, similar to SnD \cite{mai2024split}, we consider that the users privatize the token embeddings before sending them to the server.
% to protect privacy in the latent space by converting sensitive text into embeddings through a local token embedding layer on the user side. 
% However, transmitting these high-dimensional floating-point perturbed embeddings remains prohibitively costly \cite{chen2022poisson}; for instance, a sequence of 1024 tokens in Llama3-8B can incur over 16~MB of communication overhead. 
% \textcolor{blue}{Since high-dimensional token embeddings incur substantial communication cost and DP suffers from the curse of dimensionality \cite{liu2021highdimdp}, requiring stronger noise to achieve a fixed privacy level and thereby degrading utility, 
To overcome the curse of dimensionality, we first project the embeddings into a lower-dimensional latent space via a pretrained encoder, which reduces communication cost and improves the privacy–utility trade-off. Subsequently, a differentially private stochastic $n$-bit quantization mechanism is proposed to encode the latent representation on the user side, further reducing communication cost while providing local differential privacy. On the server side, the received representations are reconstructed to the original embedding dimensionality of the LLM via a paired pretrained decoder before being processed by the LLM.
% This process not only further decreases communication cost but also inherently provides local differential privacy by exploiting the stochasticity and information loss introduced by compression. 
To mitigate the utility degradation caused by privacy, instead of deploying local denoising models that may introduce prohibitive computational burden to resource-constrained users, we exploit soft prompt at the server side during inference. More specifically, we train soft prompt and prepend it to the privatized token embeddings shared by the users. It is worth mentioning that despite the application of soft prompt in various LLM tasks \cite{liu2024can,xu2024soft}, to the best of our knowledge, this is the first work that demonstrates its effectiveness in improving the privacy-utility tradeoff for LLM inference. 
% we pretrain soft prompt on the server side and prepend them to the privacy-preserving user embeddings during inference, thereby improving the LLM’s performance under DP guarantees.
% Leveraging the post-processing property of DP \cite{zhu2022postprocessing,dong2022gaussian}, once calibrated noise is injected into the original data before transmission, any subsequent operations on the privatized data cannot weaken the privacy guarantee. Consequently, the proposed framework enables the LLM server to compensate for privacy-induced utility loss while maintaining rigorous DP guarantees of user queries.

\noindent \textbf{Contributions.} In summary, this paper makes the following contributions.
\begin{itemize}
    \item We introduce a novel DP mechanism for LLM split inference, where user-side token embeddings are projected into a low-dimensional latent space and further privatized through a differentially private stochastic $n$-bit quantization mechanism, thereby providing differential privacy guarantees while significantly reducing communication overhead.
    \item We propose a lightweight privacy-preserving LLM inference framework that leverages server-side soft prompt to compensate for the utility loss caused by privacy, eliminating the need to deploy computationally intensive local denoising models on the user side.
    % fully exploiting the LLM server’s capabilities instead of depending on local post-processing models.
    % \item We introduce a projector--deprojector architecture that not only reduces communication costs via embedding compression but also provides an inherent denoising effect, leading to superior communication--privacy--utility trade-offs.
    \item Extensive experiments on both text generation and natural language understanding tasks with various state-of-the-art LLMs demonstrate that the proposed framework exhibits a better trade-off between privacy and utility than the state-of-the-art approaches in differentially private LLM inference.
\end{itemize}

\section{Related Work}
\subsection{Privacy Protection for LLMs}
With the rapid development of LLMs, concerns about privacy leakage have become increasingly prominent. Prior studies on privacy preservation mainly target the training stage, including pre-training \cite{Hoory2021LearningAE,yin-habernal-2022-privacy,Ganesh2023WhyIP}, fine-tuning \cite{Li2021LargeLM,Bu2022DifferentiallyPO,kurakin2023harnessing}, and parameter-efficient adaptation \cite{Yu2021DifferentiallyPF,Li2023PrivacyPreservingPT,tang2024privacy_preserving_icl,duan2023privacy}. In contrast, our work focuses on the inference phase, where the LLM is already pretrained and kept frozen.  

CipherGPT \cite{hou2023ciphergpt} and Ditto \cite{wu24ditto} enable inference on encrypted data via homomorphic encryption. However, their computational and communication overheads hinder deployment for large-scale LLMs. Alternatively, TextFusion \cite{zhou2022textfusion} employs dynamic token fusion, while Hide-and-Seek  \cite{Chen2023HideAS} and PrivacyRestore \cite{zeng2025privacyrestore} remove sensitive entities in the input before sending it to the server and restore the utility afterward. However, these approaches are concerned with predefined entities, leaving other elements (e.g., verbs or descriptive attributes) unprotected. 

%While effective for predefined entities, these approaches fail to protect arbitrary linguistic elements (e.g., verbs or descriptive attributes), leaving significant privacy gaps.}
% PrivacyRestore \cite{zeng2025privacyrestore} removes privacy spans in user queries and pretrains a steering vector for each span type to restore utility during inference. While effective for protecting predetermined entities, these mechanisms cannot guarantee privacy for every token in a user query, leaving some types of sensitive linguistic elements (e.g., verbs or descriptive attributes) unprotected.}

Another line of research focuses on providing token-level DP guarantees. For instance, SANTEXT \cite{yue2021differential}, CUSTEXT \cite{chen2023customized}, RANTEXT \cite{tong2025inferdpt} and CAPE \cite{wu2025cape} achieve DP by replacing words with semantically similar counterparts. Conversely, DP-Forward \cite{du2023dpforward} employs a split inference paradigm, injecting random perturbations into token embeddings prior to transmission. To further improve the privacy-utility trade-off, SnD \cite{mai2024split} introduces a local denoising module to recover perturbed outputs from remote LLMs, while similarly applying Laplacian noise to token embeddings. Although these methods achieve promising privacy–utility trade-offs, they are designed for natural language understanding tasks and do not extend well to open-ended text generation. Recently, InferDPT \cite{tong2025inferdpt} extends privacy-preserving inference to text generation tasks by perturbing the input text locally and introducing a quantized local LLM as an extraction module to refine the perturbed outputs into coherent generations aligned with the original prompts. However, deploying  LLM locally incurs substantial computation and memory costs, and the generation performance is largely constrained by the local model's capacity rather than the remote LLM service. 

\subsection{Soft Prompt}
Prompt engineering \cite{liu2023csur} has emerged as an effective paradigm to steer LLMs without modifying their parameters. By providing instructions or examples in natural language, LLMs can be guided to perform target tasks in a zero-shot or few-shot manner. However, manually designing prompts is both labor-intensive and error-prone. Consequently, a series of works aim to automate prompt optimization \cite{li2021prefix,shin2020autoprompt}.
A particularly influential approach is soft prompt tuning \cite{lester2021power}, which introduces trainable continuous embeddings (i.e., soft prompt) prepended to the input, with only these embeddings updated during training while keeping the LLM parameters frozen. During inference, the learned soft prompt is prepended to the input embeddings to adapt the model to downstream tasks. Subsequent studies have demonstrated that soft prompt achieves strong performance across diverse tasks, including natural language understanding \cite{wu2023infoprompt}, graph learning \cite{liu2024can}, and recovering the accuracy of compressed LLMs \cite{xu2024soft}. Beyond task-specific adaptation, soft prompt has also shown notable transferability across different models and task domains \cite{su2021transferability,xu2024soft}.
Nevertheless, to the best of our knowledge, no prior work has explored leveraging soft prompt to maintain the inference performance of LLMs under privacy-preserving settings.

% 介绍DP，方法中我们工作几部分：1.新的隐私衡量机制，因此需要在introduction说明chidp缺陷，以及两者trade-off对比；2.减少通信开销手段；3.提优策略，放服务器端，以及服务器实际实现可以先映射再训soft

\section{Preliminaries and Problem Setup}\label{pre}
% \subsection{Differential Privacy}
% \noindent Formally, differential privacy is defined as follows.
% \begin{definition}[$(\epsilon,\delta)$-DP \cite{dwork2006our}]\label{preliedp}
% A randomized mechanism $\mathcal{M}$ satisfies $(\epsilon,\delta)$-differentially private if for all neighboring datasets $X$, $X'$ that differ in the record of a single individual, and any event $E$,
% \begin{equation}
% P(\mathcal{M}(X)\in E) \leq e^{\epsilon} P(\mathcal{M}(X')\in E) + \delta,
% \end{equation}
% in which $\epsilon,\delta \geq 0$ are the parameters that characterize the level of differential privacy.  
% \end{definition}

\subsection{$f$-Differential Privacy}\label{prelifdp}
% \noindent We adopt the $f$-differential privacy ($f$-DP) framework, which offers a hypothesis-testing view of privacy and can be losslessly converted to $(\epsilon,\delta)$-DP. 
Given a randomized mechanism $\mathcal{M}$ and two neighboring datasets $X$ and $X'$, we consider the hypotheses
\begin{equation}\label{hypothesistestingproblem}
\begin{split}
H_0&:\text{the dataset is } X,\\
H_1&:\text{the dataset is } X'.
\end{split}
\end{equation}
Let $P$ and $P'$ denote the distributions of $\mathcal{M}(X)$ and $\mathcal{M}(X')$. 
For any rejection rule $0\le\phi\le1$, the type I and type II error rates are 
\begin{equation}
\alpha_{\phi}=\mathbb{E}_{P}[\phi],\quad \beta_{\phi}=1-\mathbb{E}_{P'}[\phi].
\end{equation}

\begin{definition}[Trade-off function~\cite{dong2022gaussian}]
For distributions $P$ and $P'$, the trade-off function $T(P,P'):[0,1]\to[0,1]$ is
\begin{equation}
T(P,P')(\alpha)=\inf\{\beta_{\phi}:\alpha_{\phi}\le\alpha\},
\end{equation}
where the infimum is over all measurable $\phi$. A larger $T(P,P')$ implies that the two distributions are harder to distinguish.
\end{definition}

\begin{definition}[$f$-DP~\cite{dong2022gaussian}]
A mechanism $\mathcal{M}$ is $f$-differentially private if 
\begin{equation}
T(\mathcal{M}(X),\mathcal{M}(X')) \ge f
\end{equation}
for all neighboring $X,X'$, meaning that an attacker cannot achieve a type II error rate smaller than $f(\alpha)$.
\end{definition}

When $T(P,P')$ is defined between two normal distributions $\mathcal{N}(0,1)$ and $\mathcal{N}(\mu,1)$, we obtain a subfamily of $f$-DP guarantees termed $\mu$-Gaussian differential privacy ($\mu$-GDP), denoted by $G_\mu(\alpha)$.

\subsection{Threat Model and Problem Formulation}
We consider the scenario where
% a split inference setting \cite{mai2024split} in which 
the user is required to submit personal queries to an LLM server maintained by a third-party provider for inference. The server is assumed to be \emph{honest-but-curious}: it honestly executes inference but may attempt to infer sensitive information from the received queries. Additionally, we assume that the server is aware of the DP mechanism, but not of the underlying private data or the realized noise.
% In addition, we account for potential eavesdroppers who may intercept transmitted data over communication channels in order to steal private information.  

\begin{figure*}[!htbp]
    \centering
    \includegraphics[width=\linewidth]{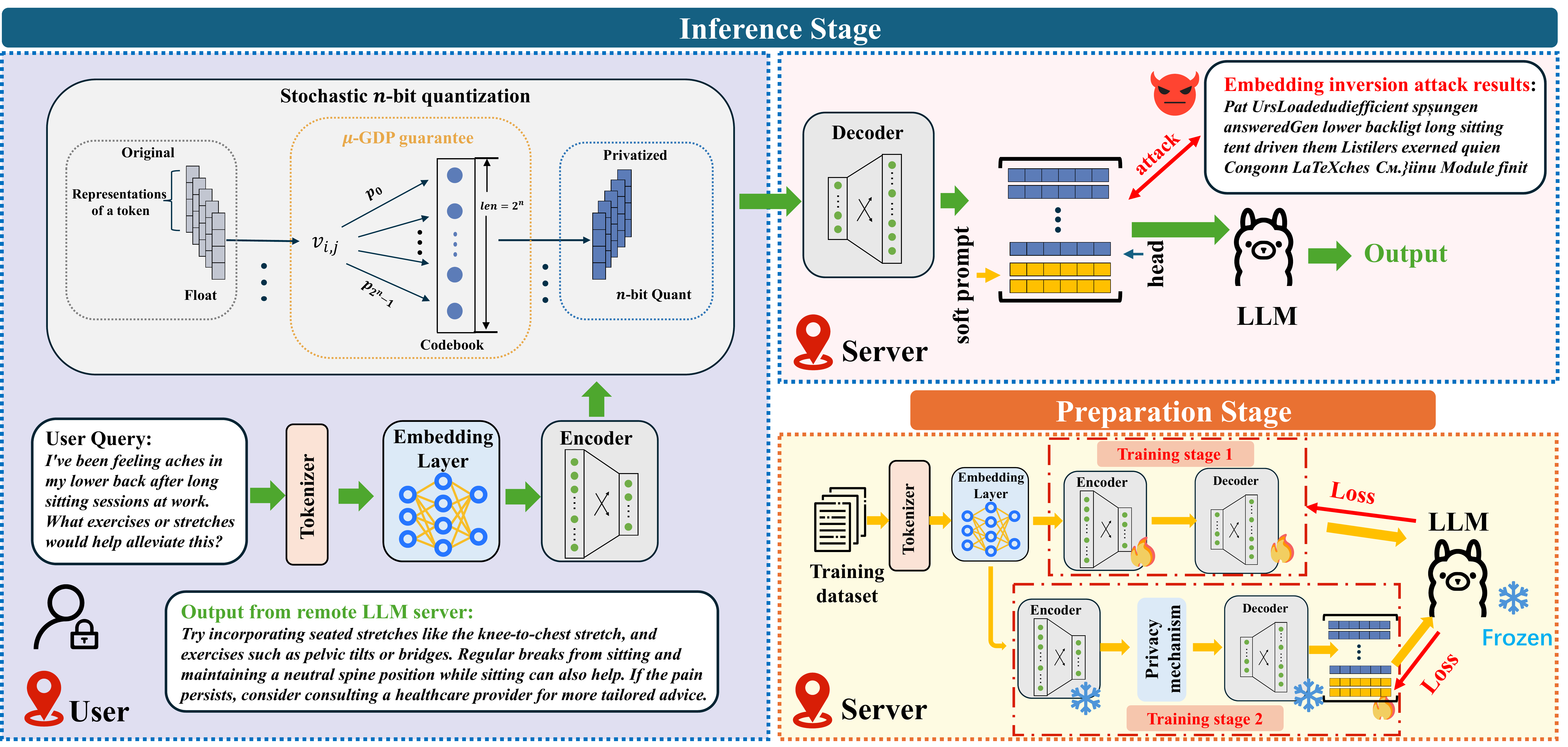}
    \vspace{-0.2in}
    \caption{An overview of the proposed differentially private and utility-preserving LLM inference framework.}
    
    \label{fig:framework}
    \vspace{-0.1in}
\end{figure*}

In this work, we partition the LLM into two components: a local encoder $G_l$ deployed on the user and a cloud encoder $G_c$ residing on the server. 
% Specifically, only the token embedding layer is retained at the user, while all subsequent layers of the LLM are executed on the server.
During inference, the user first converts the input text into latent representations through $G_l$, which are then transmitted to the server. The server continues the forward propagation through $G_c$ to generate logits or textual responses.

Formally, let $\bm{s} = [s_1, \ldots, s_T]$ denote the token sequence of length $T$ obtained from the original text query after tokenization. The local encoder $G_l$ maps $\bm{s}$ into representations $\bm{v}=[\bm{v}_1,\dots,\bm{v}_T]$, where each $\bm{v}_i=[v_{i,1},\dots,v_{i,d}]$ corresponds to a $d$-dimensional latent representation of token $s_i$. To protect the privacy in the latent space, a differential privacy mechanism $\mathcal{M}$ is applied to perturb the representations before transmission, and the privatized representations are denoted as $\hat{\bm{v}}=\mathcal{M}(\bm{v})$, which are then transmitted to the server for further processing. Given this setup, the inference task can be formulated as
\begin{equation}
    \max_\mathcal{M} \;
     p(\bm{y}\mid G_c(\mathcal{M}(G_l(\bm{s})))), \quad
    \text{s.t. } 
     \mathcal{M}\text{ satisfies } f\text{-DP},
\end{equation}
where $\bm{y}$ denotes the target output of the LLM and $\mathcal{M}$ enforces differential privacy guarantees, thereby ensuring that the released embeddings do not compromise the confidentiality of the original query.

\section{Methodology}
% \subsection{Overview}
In the following, we formally introduce the proposed \textbf{D}ifferentially private and communication \textbf{E}fficient \textbf{L}LM split inference (\textbf{DEL}) framework. The overall pipeline of the proposed framework is illustrated in Figure~\ref{fig:framework}.

% \textbf{Privacy module:} Ensures privacy protection. Given the embedding vector of the user’s original text, this module perturbs it locally on the user side before transmission, providing a formal Local Differential Privacy (LDP) guarantee.

% \textbf{Utility module:} Preserves task utility. Since privatized inputs to LLMs inevitably lead to degraded outputs, this module leverages server-side soft prompt to enhance the performance of the LLMs under perturbed inputs.

% \textbf{Compression module:} Reduces communication cost. To relieve the high dimensionality and consequent extensive communication cost of embeddings, this module reduces the dimension of embedding vectors, thereby lowering bandwidth requirements.

\subsection{Privacy Protection mechanism}\label{privacy_section}
Given a user query, DEL perturbs it locally in the latent space on the user side, providing the local DP guarantee. Specifically, the user's token sequence $\bm{s}$ is first mapped to token embeddings $\bm{x}=[\bm{x}_1,\dots,\bm{x}_T]$ through the embedding layer of the LLM, where each $\bm{x}_i \in \mathbb{R}^b$ denotes the $b$-dimensional embedding of token $s_i$. 

The Gaussian mechanism \cite{abadi2016deep} is the most widely adopted DP mechanism owing to its tight privacy guarantees and favorable composability properties \cite{abadi2016deep}. It satisfies $\mu$-GDP and is near order-optimal across most parameter regimes \cite{chen2023privacy}. Formally, given vector $\bm{x}_i$, the Gaussian mechanism outputs
\begin{equation}
    \mathcal{M}^{\text{Gauss}}(\bm{x}_{i}) \;=\; \text{clip}(\bm{x}_{i}, C) + \bm{n}, 
    \quad \bm{n} \sim \mathcal{N}\!\left(0, \sigma^2\right),
\end{equation}
which satisfies $\frac{2C}{\sigma}$-GDP and $C$ denotes the $\ell_2$-norm clipping bound.

Nevertheless, the high dimensionality and continuous nature of the Gaussian mechanism lead to poor communication efficiency. Moreover, as discussed in \cite{liu2021highdimdp}, differential privacy suffers from the curse of dimensionality. Specifically, the amount of noise required to achieve a fixed privacy level grows rapidly with the data dimension, leading to severe utility degradation in high-dimensional settings. To address these issues, we first introduce a dimensionality-reduction module consisting of an encoder–decoder pair $(g_e, g_d)$. The encoder $g_e$ projects $\bm{x}$ to $\bm{v}$ with each $\bm{v}_i\in\mathbb{R}^d$ ($d<b$), yielding the latent representations $\bm{v}=g_e(\bm{x})$. The decoder $g_d$ reconstructs the embeddings from the latent representations as $\hat{\bm{x}}=g_d(\bm{v})$ on the server for LLM generation, where each $\hat{\bm{x}}_i\in\mathbb{R}^b$. 
To derive both components, we pretrain them on the server side 
% as illustrated in the Training stage 1 in Figure~\ref{fig:framework}, 
and after which the encoder $g_e$ is distributed to users. 
The learning objective is defined as
\begin{align}
    \min_{g_e,g_d} \sum_{i=1}^T 
    -\log \Pr\!\Big[ \bm{x}_{i+1} \,\Big|\,
    & G_c(g_d(g_e(\bm{x}_1))), \dots, \notag \\
    & G_c(g_d(g_e(\bm{x}_{i}))) \Big].
\end{align}

Furthermore, we propose a stochastic $n$-bit quantization mechanism, which explicitly incorporates the privacy contribution of the stochastic quantization process and ensures that the overall procedure remains $\mu$-GDP while achieving $n$-bit quantization.
\begin{definition}[Stochastic $n$-bit Quantization]\label{define_sto}
Let $\bm{v}_i = [v_{i,1}, \dots, v_{i,d}] \in [-c, c]^d$ be a $d$-dimensional vector,
% , where $c > 0$ is the clipping bound, 
and let $A \ge c$ be a scaling parameter that controls the privacy budget.  
The stochastic $n$-bit quantization mechanism $\mathcal{M}^{\mathrm{sto}}(\cdot; A, n)$ maps each coordinate $v_{i,j}$ to one of $2^n$ uniformly spaced levels in $[-A, A]$.
% :
% \[
% \mathcal{Q}_n = \left\{q_k = \tfrac{2k - (2^n - 1)}{2^n - 1}A \,\middle|\, k = 0, \dots, 2^n - 1 \right\}.
% \]
For each coordinate, define $u = 2^n - 1$ and $p(v_{i,j}) = (A + v_{i,j}) / (2A)$.  
Then an integer $K \in \{0, \dots, u\}$ is randomly drawn from the binomial distribution
\begin{equation}
   K \sim \mathrm{Binomial}\big(u,\, p(v_{i,j})\big), 
\end{equation}
and the quantized output is
\begin{equation}
\mathcal{M}^{\mathrm{sto}}(v_{i,j}; A, n) = \tfrac{2K - (2^n - 1)}{2^n - 1}A.
\end{equation}
Applying this mapping independently to each coordinate yields the quantized vector
$\mathcal{M}^{\mathrm{sto}}(\bm{v}_i; A, n)$.
\end{definition}

Intuitively, for any two distinct coordinates in $[-c,c]$, the stochastic $n$-bit quantization can map them, each with nonzero probability, to the same level in $[-A,A]$, which makes them indistinguishable from the output alone. Moreover, when $A\!\gg\!c$, the binomial parameter satisfies $p(v_{i,j})\!\approx\!\tfrac{1}{2}$ uniformly over $v_{i,j}\in[-c,c]$, so the output distribution becomes nearly input-agnostic, thereby corresponding to stronger privacy.

% Theorem \label{theorem_sto} below gives the variance and privacy guarantee of the proposed stochastic $n$-bit quantization, and the proofs can be found in Appendix \ref{theorem_sto}.

\begin{theorem}\label{theorem_sto}
For a given vector $\bm{v}_i = [v_{i,1}, \dots, v_{i,d}] \in [-c, c]^d$, 
$\hat{\bm{v}}_i=\mathcal{M}^{\text{sto}}(\bm{v}_i; A, n)$ is an unbiased estimator of $\bm{v}_i$ with the variance $\operatorname{Var}(\mathcal{M}^{\text{sto}}(\bm{v}_{i}; A, n))
=\frac{dA^2-\|\bm{v}_i\|_2^2}{2^n-1}$. 
Moreover, the mechanism is $f^\text{sto}(\alpha)$-DP with
\begin{equation}
        G_{\mu}(\alpha+\gamma)-\gamma \leq f^\text{sto}(\alpha) \leq G_{\mu}(\alpha-\gamma)+\gamma,
\end{equation}
in which
\begin{equation}
\mu = \frac{2\sqrt{(2^n-1)d}c}{\sqrt{A^2-c^2}},
\end{equation}
\begin{equation}
\gamma = \frac{0.56\left[\frac{A-c}{2A}\left|1+\frac{c}{A}\right|^3+\frac{A+c}{2A}\left|1-\frac{c}{A}\right|^3\right]}{(1-\frac{c^2}{A^2})^{3/2}\sqrt{(2^n-1)d}}.
\end{equation}

% Moreover, the mechanism's privacy guarantee asymptotically approaches $\mu$-GDP as $d$ increases, where $\mu = \frac{2\sqrt{(2^n - 1)d}\,c}{\sqrt{A^2 - c^2}}$.
\end{theorem}

\begin{remark}\label{remark_ana}
    Let $\bm{v}_i \in \mathbb{R}^d$ be processed by both the Gaussian mechanism and the stochastic $n$-bit quantization mechanism, where $\left|v_{i,j}\right|\leq c=\frac{C}{\sqrt{d}}$ for all $j$. Under the conditions $A^2 - c^2 = (2^n - 1)\sigma^2$ and $(1 - \frac{c^2}{A^2})(2^n - 1)d \gg 1.12$, the $\mu$ in Theorem~\ref{theorem_sto} coincides with that of the Gaussian mechanism, and $\gamma$ in Theorem~\ref{theorem_sto} vanishes asymptotically. In this case, we have $\operatorname{Var}(\mathcal{M}^{\text{Gauss}}(\bm{v}_{i}))=\sigma^2d$ and $\operatorname{Var}(\mathcal{M}^{\text{sto}}(\bm{v}_{i}; A, n))-\operatorname{Var}(\mathcal{M}^{\text{Gauss}}(\bm{v}_{i}))=\frac{c^2d-\|\bm{v_i}\|_2^2}{2^n-1}$. Specifically, when $v_{i,j}\in \{-c, c\}$, the stochastic $n$-bit quantization mechanism demonstrates the same $f$-DP privacy guarantee and variance as that of Gaussian mechanism, i.e., the improvement in communication efficiency is obtained for free.
\end{remark}
% The detailed proofs are in Appendix \ref{theorem_sto}.

\subsection{Utility Preserving Mechanism}
Since perturbing the input embeddings inevitably degrades the inference performance of LLM, we alleviate this issue by adapting soft prompt. Specifically, given a privatized latent representation $\hat{\bm{v}}$, the corresponding privatized embedding reconstructed by the server-side decoder is denoted as $\hat{\bm{x}}$. Together with the soft prompt $\bm{E} = [\bm{e}_1, \dots, \bm{e}_r]$, where each $\bm{e}_j \in \mathbb{R}^b$, the final input to the server-side model is constructed as the concatenation $[\bm{E}, \hat{\bm{x}}]$. The output of the cloud encoder at the server is then represented as $G_c([\bm{E},\hat{\bm{x}}])$. 
% By the post-processing property of differential privacy \cite{dong2022gaussian}, the insertion of the soft prompt does not compromise the DP guarantee provided for the user embedding.

% For tuning soft prompt, 
% we tuning $\bm{E}$ by
% incorporating the corresponding privacy mechanism with the frozen LLM and frozen pretrained dimentionality reduction module on the server side.
For soft prompt tuning, we optimize only the soft prompt $\bm{E}$ while keeping the LLM and the pretrained encoder–decoder pair $(g_e, g_d)$ frozen, and perform the optimization under the corresponding privacy mechanism on the server side.
The training objective is to maximize the conditional likelihood of the ground-truth sequence given the privatized embeddings $\hat{\bm{x}}$ and the prepended soft prompt. 
Formally, the optimization problem can be written as
\begin{equation}
    \min_{\bm{E}} \sum_{i=1}^T 
    -\log \Pr\!\left[\bm{x}_{i+1} \,\middle|\,
    \bm{e}_1,\dots,\bm{e}_r,\,\hat{\bm{x}}_1,\dots,\hat{\bm{x}}_i\right].
\end{equation}
Through this optimization, the learned prompt becomes privacy-aware, adapting the model’s behavior to the distributional shift introduced by the privatized embeddings. 
This enables the server-side LLM to effectively compensate for the degradation caused by privacy during inference. 
As detailed in Section~\ref{exp}, we consider employing either the training data of the target dataset or publicly available data as the source for soft prompt tuning at the server, and the results, as presented in Table~\ref{tab:soft_mu_transfer}, demonstrate that incorporating the soft prompt can substantially enhance the model’s performance across various privacy levels.

\begin{table*}[!htbp]
\centering
\caption{Comparison of open-ended text generation performance across different ASRs. All results are reported using the COH score.}
\label{tab:results_transposed}
\scriptsize
\setlength{\tabcolsep}{3.5pt}
\renewcommand{\arraystretch}{1.15}
\begin{tabular}{l|c|ccc|c|ccc}
\hline
\multirow{2}{*}[-0.5ex]{\makecell[c]{Method}} & 
% \multicolumn{4}{c|}{\textbf{InferDPT}} & 
\multirow{2}{*}[-0.5ex]{\makecell[c]{\textbf{RANTEXT}}} & \multicolumn{3}{c|}{\textbf{InferDPT}} &
\multirow{2}{*}[-0.5ex]{\makecell[c]{\textbf{DEL}}} &
\multicolumn{3}{c}{\textbf{DEL + Local Model}} \\
\cline{3-5} \cline{7-9}
 &
 % & \makecell{w/o local \\ model} 
 & \makecell{with local \\ DialoGPT} & \makecell{with local \\ OPT-1.3B} & \makecell{with local \\ Llama2-4bit} &  & DialoGPT & OPT-1.3B & Llama2-4bit \\
\hline
\multicolumn{9}{c}{\textbf{Llama3-8B + PTB}} \\
\hline
ASR = 0.02 & 0.517 & 0.581 & 0.565 & 0.624 & 0.570 & 0.587 & 0.558 & \textbf{0.639} \\
ASR = 0.10 & 0.520 & 0.585 & 0.567 & 0.615 & \textbf{0.636} & 0.592 & 0.620 & 0.636 \\
ASR = 0.15 & 0.522 & 0.584 & 0.566 & 0.617 & \textbf{0.650} & 0.592 & 0.618 & 0.636 \\
ASR = 0.20 & 0.526 & 0.583 & 0.561 & 0.628 & \textbf{0.661} & 0.591 & 0.633 & 0.642 \\
\hline
\multicolumn{9}{c}{\textbf{Llama3-8B + Wikitext-2}} \\
\hline
ASR = 0.02 & 0.477 & 0.515 & 0.665 & \textbf{0.742} & 0.562 & 0.527 & 0.602 & 0.736 \\
ASR = 0.10 & 0.492 & 0.517 & 0.639 & 0.741 & 0.720 & 0.522 & 0.653 & \textbf{0.747} \\
ASR = 0.15 & 0.493 & 0.518 & 0.648 & 0.748 & \textbf{0.770} & 0.522 & 0.683 & 0.765 \\
ASR = 0.20 & 0.492 & 0.519 & 0.658 & 0.735 & \textbf{0.783} & 0.523 & 0.692 & 0.766 \\
\hline
\multicolumn{9}{c}{\textbf{Qwen2.5-7B + PTB}} \\
\hline
ASR = 0.02 & 0.517 & 0.570 & 0.558 & 0.623 & 0.590 & 0.589 & 0.598 & \textbf{0.640} \\
ASR = 0.10 & 0.521 & 0.575 & 0.568 & 0.629 & \textbf{0.655} & 0.592 & 0.628 & 0.645 \\
ASR = 0.15 & 0.526 & 0.574 & 0.561 & 0.617 & \textbf{0.661} & 0.592 & 0.625 & 0.639 \\
ASR = 0.20 & 0.525 & 0.571 & 0.565 & 0.629 & \textbf{0.677} & 0.593 & 0.630 & 0.637 \\
\hline
\multicolumn{9}{c}{\textbf{Qwen2.5-7B + Wikitext-2}} \\
\hline
ASR = 0.02 & 0.486 & 0.515 & 0.649 & \textbf{0.752} & 0.690 & 0.527 & 0.642 & 0.732 \\
ASR = 0.10 & 0.489 & 0.514 & 0.638 & 0.751 & \textbf{0.781} & 0.521 & 0.697 & 0.764 \\
ASR = 0.15 & 0.493 & 0.516 & 0.640 & 0.756 & \textbf{0.790} & 0.525 & 0.698 & 0.770 \\
ASR = 0.20 & 0.493 & 0.516 & 0.642 & 0.739 & \textbf{0.800} & 0.523 & 0.700 & 0.778 \\
\hline
\end{tabular}
% \vspace{-0.1in}
\end{table*}

\section{Experiments}\label{exp}
\subsection{Experiments Setup}
\textbf{Datasets and Models:} We evaluate the proposed method on four widely used language generation benchmarks: the C4 web corpus from Common Crawl \cite{Raffel2020T5}, WikiText-2 \cite{Merity2017PointerSentinel}, Penn Treebank (PTB) \cite{marcus-etal-1994-penn}, CNN/Daily Mail \cite{hermann2015teaching}, and two natural language understanding (NLU) downstream tasks: Quora Question Pairs (QQP) \cite{quora2017qqp} and MSR Paraphrase Corpus (MRPC) \cite{dolan-brockett-2005-automatically}. 
For LLMs, we adopt Llama3 \cite{touvron2023}, Qwen2.5 \cite{qwen2025qwen25}, DeepSeek-MoE \cite{Dai2024DeepSeekMoE}, Pangu-Embedded \cite{Chen2025PanguEmbedded} for open-ended text generation tasks and BERT \cite{devlin2019bert} for NLU tasks.

% Tweet Eval \cite{barbieri2020tweeteval}, 

\textbf{Metrics:} For performance evaluation, we report Perplexity (PPL), which quantifies how well the model predicts the next token, and Coherence (COH) \cite{tong2025inferdpt}, which measures the semantic consistency between the generated and target texts based on sentence embeddings for the open-ended text generation task. More details are given in Appendix \ref{metric_app}.
% Perplexity is a standard intrinsic metric for language modeling, which measures how well the model predicts a given sequence. 
% Formally, for a sequence $\bm{s}=(s_1,\dots,s_T)$, the perplexity is defined as
% \begin{equation}
%     \mathrm{PPL}(\bm{s})=\exp\left(-\frac{1}{T}\sum_{t=1}^{T}\log p(s_t|s_{<t})\right),
% \end{equation}
% where $p(s_t|s_{<t})$ denotes the probability assigned by the model to the next token $s_t$ given the prefix $s_{<t}$. 
% Lower perplexity indicates that the model assigns a higher likelihood to the ground-truth sequence, thus reflecting stronger language modeling ability.  
% Coherence evaluates the semantic consistency between the user prompt text $Doc$ and the generated output text of LLM $Gen$ by computing the cosine similarity of their embeddings:
% \begin{equation}
%     COH(Doc,Gen)=\frac{\text{SimCSE}(Doc)\cdot \text{SimCSE}(Gen)}{\|\text{SimCSE}(Doc)\|\cdot \|\text{SimCSE}(Gen)\|},
% \end{equation}
% where $\text{SimCSE}(\cdot)$ denotes the pretrained sentence embedding model \cite{gao-etal-2021-simcse}.
For NLU tasks, we follow \cite{mai2024split} and report accuracy (Acc) and area under the ROC curve (AUC) as evaluation metrics.

% To evaluate the privacy mechanism, we consider both the widely adopted white-box embedding inversion attack and input inference attack \cite{yue2021differential,mai2024split,tong2025inferdpt}. The embedding inversion attack computes the cosine similarity between the embedding of each token in the perturbed input and all token embeddings in the vocabulary, and predicts the token with the highest similarity as the recovered one. The input inference attack sequentially replacing the perturbed token with the special token [MASK] and employ BERT \cite{devlin-etal-2019-bert} to predict it.  The attack success rate (ASR) is reported as the metric for privacy evaluation. We find that embedding inversion attack always exhibits higher ASR than input inference attack and defaultly report the ASR of embedding inversion attack in the following. More details can be found in the Appendix.

To evaluate privacy, we consider both embedding inversion and input inference attacks \cite{mai2024split,tong2025inferdpt}. The embedding inversion attack recovers tokens by identifying the vocabulary entry with the highest similarity to the perturbed embedding. The input inference attack works by sequentially replacing perturbed tokens with [MASK] and employing BERT to predict them. We report the ASR of the embedding inversion attack by default, as it consistently outperforms the input inference attack (see Appendix \ref{metric_app} for detailed comparisons).

% and train the soft prompt on the training splits of the downstream evaluation datasets under specified privacy budgets.

\textbf{Baselines:} To demonstrate the effectiveness of the proposed method, we compare it with representative open-source baselines. We primarily focus on the open-ended text generation task, benchmarking against the state-of-the-art RANTEXT and InferDPT \cite{tong2025inferdpt}. Moreover, to assess performance on NLU tasks, we compare it with the best-performing baseline SnD \cite{mai2024split}.

\textbf{Other Hyperparameters:} For the open-ended text generation task, the length of the soft prompt $r$ is set to 100, while for the NLU tasks, it is set to 20. 
We set the bound $c$ to 0.05 and the reduced dimension $d$ to $b/32$. 
% The soft prompt is optimized using AdamW with learning rates of $1\times10^{-3}$.
The soft prompt and the encoder-decoder module are optimized using AdamW with learning rates of $1\times10^{-3}$ and $1\times10^{-4}$, respectively.

\textbf{Pretraining Details:}
In the experiments, we first pretrain the encoder and decoder components of the privacy mechanism on the public C4 dataset and freeze them during subsequent soft prompt tuning and inference.
For soft prompt tuning, we utilize the training data for each dataset.
To account for scenarios where no training data are publicly available on the server side, we also conduct experiments where the soft prompt is trained on the public C4 dataset and examine its transferability.

\subsection{Performance on Open-ended Text Generation Tasks} 
\begin{figure}[t]
    \centering
    % \vspace{-0.5in}
    \begin{subfigure}[b]{0.23\textwidth}
        \centering
        \includegraphics[width=1\textwidth]{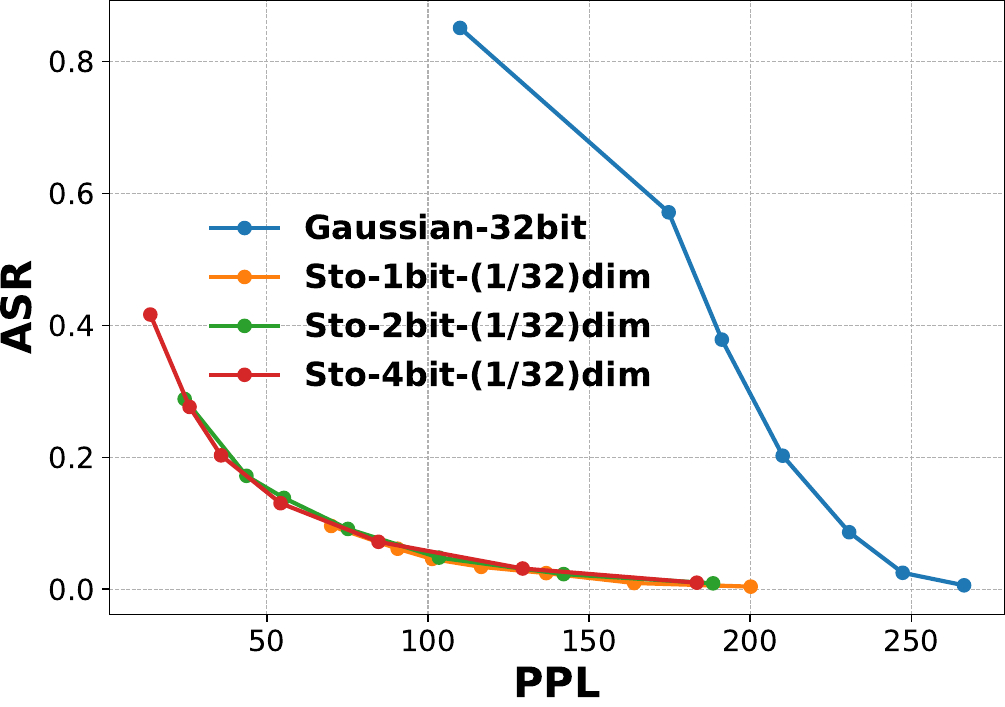}
        % \vspace{-0.5in}
        % \caption{FedAvg}
        \label{fig:fig11}
    \end{subfigure}
    \hfill
    % \vspace{-0.2in}
    \begin{subfigure}[b]{0.23\textwidth}
        \centering
        \includegraphics[width=1\textwidth]{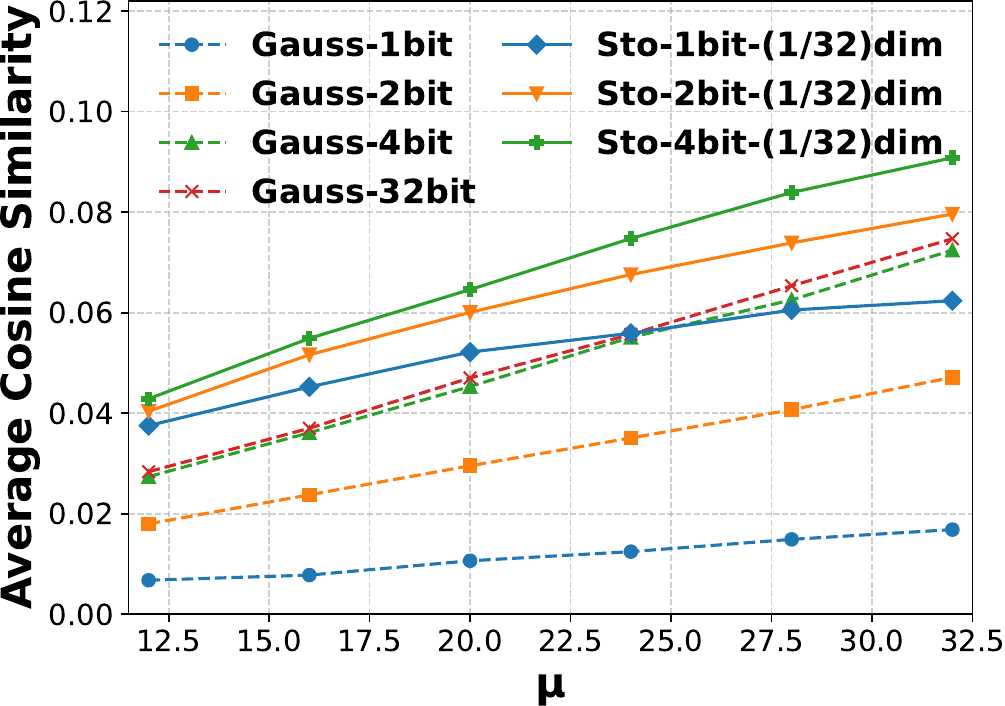}
        % \vspace{-0.5in}
        % \caption{FedSAM}
        \label{fig:fig12}
    \end{subfigure}
    \vspace{-0.2in}
    \caption{Comparison between the Proposed Method and the Gaussian Mechanism on the PTB Dataset using Qwen2.5-7B.}
    \label{fig:fig1}
    \vspace{-0.1in}
\end{figure}

\begin{table}[t]
\centering
% \vspace{-0.1in}
\caption{Performance of LLMs without privacy protection.}
\label{tab:ptb_wikitext}
\scriptsize   % 可以换成 \footnotesize 或 \small
\setlength{\tabcolsep}{6pt}   % 调整列间距，默认6pt，可以再缩小
\renewcommand{\arraystretch}{1.1} % 调整行高
\begin{tabular}{l|cc}
\hline
\multirow{2}{*}{Model} & \multicolumn{2}{c}{Accuracy} \\
\cline{2-3}
 & PTB & Wikitext-2 \\
\hline
Llama3-8B      & 0.651 & 0.773 \\
Qwen2.5-7B     & 0.668 & 0.774 \\
Qwen2.5-72B    & 0.680 & 0.803 \\
DeepSeek-MoE-16B & 0.651 & 0.763 \\
Pangu-Embedded-7B & 0.626 & 0.766 \\
DialoGPT       & 0.588 & 0.520 \\
OPT-1.3B-chat  & 0.599 & 0.695 \\
Llama2-7B-chat & 0.635 & 0.772 \\
\hline
\end{tabular}
\vspace{-0.2in}
\end{table}
We begin by comparing the proposed method with the classical Gaussian mechanism, which is known to be near order-optimal in many parameter regimes. In the Gaussian mechanism, additive noise is directly applied to the token embeddings.
For a fair comparison, we also pretrain soft prompt for the Gaussian mechanism and then evaluate the PPL of both methods on the corresponding evaluation datasets.
\begin{table*}[t]
\centering
\caption{Accuracy (Acc) and AUC comparison on MRPC under different attack success rates.}
\label{table1}
\scriptsize
\setlength{\tabcolsep}{6pt}
\renewcommand{\arraystretch}{1.15}
\begin{tabular}{l|ccc|ccc}
\hline
\multirow{2}{*}{Method} & \multicolumn{3}{c|}{Acc} & \multicolumn{3}{c}{AUC} \\
\cline{2-7}
 & ASR=0.02 & ASR=0.1 & ASR=0.4 & ASR=0.02 & ASR=0.1 & ASR=0.4 \\
\hline
Metric-DP+denoise & \textbf{0.711 ± 0.018} & 0.703 ± 0.006 & 0.694 ± 0.008 & 0.598 ± 0.020 & 0.586 ± 0.019 & 0.585 ± 0.016 \\
Stochastic 4-bit+denoise & 0.710 ± 0.011 & \textbf{0.708 ± 0.007} & 0.695 ± 0.005 & 0.590 ± 0.010 & \textbf{0.595 ± 0.017} & 0.574 ± 0.019 \\
Stochastic 2-bit+denoise & 0.703 ± 0.010 & 0.696 ± 0.006 & 0.703 ± 0.006 & 0.584 ± 0.012 & 0.584 ± 0.012 & 0.579 ± 0.007 \\
Stochastic 1-bit+denoise & 0.701 ± 0.006 & 0.698 ± 0.009 & \textbf{0.703 ± 0.008} & \textbf{0.594 ± 0.006} & 0.578 ± 0.010 & \textbf{0.586 ± 0.001} \\
\hline
\end{tabular}
\vspace{-0.1in}
\end{table*}
\begin{table*}[t]
\centering
\caption{Accuracy (Acc) and AUC comparison on QQP under different attack success rates.}
\label{table2}
\scriptsize
\setlength{\tabcolsep}{6pt}
\renewcommand{\arraystretch}{1.15}
\begin{tabular}{l|ccc|ccc}
\hline
\multirow{2}{*}{Method} & \multicolumn{3}{c|}{Acc} & \multicolumn{3}{c}{AUC} \\
\cline{2-7}
 & ASR=0.02 & ASR=0.1 & ASR=0.4 & ASR=0.02 & ASR=0.1 & ASR=0.4 \\
\hline
Metric-DP+denoise & 0.706 ± 0.016 & 0.698 ± 0.005 & 0.694 ± 0.008 & 0.673 ± 0.019 & 0.659 ± 0.015 & 0.646 ± 0.016 \\
Stochastic 4-bit+denoise+soft & \textbf{0.737 ± 0.008} & \textbf{0.725 ± 0.011} & 0.721 ± 0.010 & \textbf{0.716 ± 0.011} & 0.698 ± 0.007 & 0.694 ± 0.011 \\
Stochastic 2-bit+denoise+soft & 0.727 ± 0.009 & 0.724 ± 0.004 & \textbf{0.738 ± 0.005} & 0.695 ± 0.009 & \textbf{0.706 ± 0.007} & \textbf{0.712 ± 0.013} \\
Stochastic 1-bit+denoise+soft & 0.712 ± 0.015 & 0.724 ± 0.017 & 0.729 ± 0.002 & 0.686 ± 0.012 & 0.694 ± 0.010 & 0.704 ± 0.005 \\
\hline
\end{tabular}
\vspace{-0.2in}
\end{table*}
% , datasets, and LLM backbones, comparing against the Gaussian mechanism, which is order-optimal in many parameter regimes. 
% For different datasets, privacy mechanisms, and privacy guarantees, 
% we pretrain separate soft prompt for both our proposed mechanism with the dimensionality reduction module 
% and the classical Gaussian mechanism, which directly adds Gaussian noise to token embeddings before transmission. 
As shown in the left subfigure of Figure~\ref{fig:fig1}, compared with the classical Gaussian mechanism, the proposed method achieves lower PPL under the same privacy guarantee, where privacy is quantified by ASR of the embedding inversion attack. 
This corresponds to better language modeling performance and demonstrates a more favorable privacy–utility trade-off. 
% Compared with the analysis in Remark \ref{remark_ana}, which shows that stochastic $n$-bit quantization can only approximate the variance of the Gaussian mechanism given the same vector, the superior performance here can be attributed to the use of an encoder that reduces the embedding dimensionality on the server side. This dimensionality reduction effectively mitigates the curse of dimensionality in differential privacy, as discussed in Section \ref{privacy_section}.
In the right subfigure, we incorporate an unbiased stochastic quantizer~\cite{Alistarh2017QSGD} into the Gaussian mechanism to discretize its floating-point outputs. Under the same $\mu$-GDP budget and quantization levels, the proposed method consistently achieves higher cosine similarity between the privatized and original token embeddings than the Gaussian mechanism, indicating better preservation of semantic consistency without compromising privacy.
Moreover, the proposed method substantially reduces communication overhead. Whereas the Gaussian mechanism transmits the 3584-dimensional embeddings for Qwen2.5-7B, the proposed method compresses the embedding dimension by a factor of 32.
% , and the stochastic $n$-bit compressor further reduces the bit width per coordinate to 4, 2, or even 1 bit.

% To the best of our knowledge, this work is the first to adopt $f$-DP to formally guarantee the privacy of user queries in LLM inference. 

Table~\ref{tab:results_transposed} presents the comparison of open-ended text generation tasks.
% , the state-of-the-art approach for privacy-preserving open-ended text generation. 
While RANTEXT and InferDPT adopt the classical $\epsilon$-DP formulation and the proposed method is established under the $f$-DP framework, we ensure comparability by calibrating both mechanisms to the same target ASR. We achieve this by performing a binary search to identify the corresponding $\mu$ (for the proposed method) or $\epsilon$ (for RANTEXT and InferDPT) that yields the desired ASR on the evaluation dataset, and subsequently evaluate the generation quality under the resulting privacy budgets. 
% Following InferDPT, we adopt coherence as the evaluation metric, which quantifies semantic similarity between generated responses and reference answers using a pretrained sentence embedding model~\cite{gao-etal-2021-simcse}. 
To comprehensively assess InferDPT, we incorporate various local models with different capacities, including DialoGPT \cite{zhang-etal-2020-dialogpt}, OPT-1.3B \cite{zhang2022opt}, and Llama2-7B-4bit \cite{frantar2022gptq}. As shown in the table, the performance of InferDPT heavily depends on the capability of the local model, and its coherence scores remain nearly unchanged across different ASR levels despite the varying privacy budgets. For reference, Table~\ref{tab:ptb_wikitext} reports the performance of the same models without perturbation on the input, where InferDPT with any local model achieves performance close to using the local model alone. We attribute this to the instruction prompt used in InferDPT for local extraction (see Appendix \ref{inferspt_prompt}), which includes the original unperturbed prompt so that the local model can effectively bypass the noisy generation of the remote server and rely solely on its own capacity, thereby diminishing the contribution of the remote LLM. In contrast, the proposed mechanism demonstrates a more favorable privacy-utility trade-off. Under weaker privacy requirements, the performance approaches and even surpasses that of the remote LLMs reported in Table~\ref{tab:ptb_wikitext}, 
% and under certain low-privacy settings it even surpasses them, 
highlighting that the proposed method not only preserves user privacy but also fully exploits the capability of remote LLMs, instead of relying on local models. 
To ensure a fair comparison, we further conduct experiments where the proposed mechanism is augmented with local models. Specifically, the generation outputs from the LLM server are fed into the local LLM using the same instruction prompt as in InferDPT, enabling the local model to denoise the generated text. The corresponding results are presented in the right part of Table~\ref{tab:results_transposed}. It can be observed that the proposed method outperforms InferDPT when using the same local model in most cases, further demonstrating its effectiveness. Notably, when the privacy requirement is less stringent (i.e., at higher ASR), DEL outperforms its variant that uses local models. This improvement stems from the stronger capability of the server-side LLM, which DEL can effectively leverage. In contrast, when local models are used, overall task performance is hindered. Additional results on the MoE-architecture model (i.e., DeepSeek-MoE), the larger model (i.e., Qwen2.5-72B), and the Dual-System architecture model (i.e., Pangu-Embedded-7B) are presented in Table~\ref{tab:results_appendix}, which consistently demonstrate that the proposed mechanism achieves a more favorable privacy–utility trade-off.

\subsection{Performance on NLU Tasks} 
While the proposed framework illustrated in Figure~\ref{fig:framework} is primarily designed for open-ended text generation, we evaluate the effectiveness of its key components within the SnD framework on NLU tasks for a more extensive comparison. A detailed description of the SnD setup for NLU tasks is provided in Appendix~\ref{snd_detail}. We begin by comparing the proposed stochastic $n$-bit quantization mechanism with the metric-DP mechanism adopted in SnD. Following the same procedure as described earlier, we fix the ASR of the embedding inversion attack and perform a binary search to determine the corresponding privacy budget for each mechanism, under which the denoising model is subsequently trained.
As shown in Table~\ref{table1}, the proposed privacy mechanism achieves performance comparable to that of SnD while substantially reducing communication overhead. Specifically, the metric-DP mechanism in SnD (implemented via a Laplacian mechanism) demands the user to transmit floating-point embeddings and requires 32 bits per coordinate, whereas the proposed stochastic compressor reduces this requirement to 4, 2, or even 1 bit per coordinate with only marginal degradation in utility.
Furthermore, by incorporating soft prompt to better align the perturbed embeddings before feeding them into the LLM, DEL consistently outperforms SnD in both accuracy and AUC, as reported in Table~\ref{table2}. These results collectively demonstrate that the proposed mechanisms achieve a more favorable trade-off among privacy, communication efficiency, and utility compared to SnD.

\subsection{Transferability of Soft Prompt}

\begin{table}[t]
\centering
\caption{Additional comparison results of open-ended text generation performance under different ASRs. 
All results are reported using the COH score.}
\label{tab:results_appendix}
\scriptsize
\setlength{\tabcolsep}{5pt}
\renewcommand{\arraystretch}{1.15}
\begin{tabular}{l|c|c|c|c}
\hline
\multirow{2}{*}[-0.5ex]{ASR} 
& \multirow{2}{*}[-0.5ex]{\textbf{RANTEXT}}
& \multicolumn{1}{c|}{\textbf{InferDPT}} 
& \multirow{2}{*}[-0.5ex]{\textbf{DEL}}
& \multicolumn{1}{c}{\textbf{DEL}} \\ 
% \cline{3-5}
% & \makecell{w/o local \\ model} 
&
& \makecell{with local \\ OPT-1.3B} 
% & \makecell{w/o local \\ model} 
&
& \makecell{with local \\ OPT-1.3B}  \\ 
\hline
% \multicolumn{5}{c}{\textbf{Pangu-Embedded-7B + PTB}} \\ 
% \hline
% 0.02 & 0.510 & 0.549 & 0.553 & \textbf{0.571} \\
% 0.10 & 0.515 & 0.544 & \textbf{0.633} & 0.614 \\
% 0.15 & 0.516 & 0.552 & \textbf{0.635} & 0.617 \\
% 0.20 & 0.516 & 0.542 & \textbf{0.642} & 0.622 \\
% \midrule
\multicolumn{5}{c}{\textbf{Pangu-Embedded-7B + Wikitext-2}} \\ 
\hline
0.02 & 0.490 & 0.593 & \textbf{0.615} & 0.611 \\
0.10 & 0.494 & 0.578 & \textbf{0.697} & 0.646 \\
0.15 & 0.501 & 0.582 & \textbf{0.727} & 0.659 \\
0.20 & 0.504 & 0.585 & \textbf{0.742} & 0.662 \\
\midrule
% \multicolumn{5}{c}{\textbf{DeepSeek-MoE-16B + PTB}} \\ 
% \hline
% 0.02 & 0.524 & 0.575 & 0.566 & 0.569 \\
% 0.10 & 0.518 & 0.580 & 0.608 & 0.596 \\
% 0.15 & 0.522 & 0.578 & 0.634 & 0.612 \\
% 0.20 & 0.523 & 0.580 & 0.664 & 0.627 \\
% \midrule
\multicolumn{5}{c}{\textbf{DeepSeek-MoE-16B + Wikitext-2}} \\ 
\hline
0.02 & 0.470 & \textbf{0.664} & 0.586 & 0.632 \\
0.10 & 0.479 & 0.666 & \textbf{0.705} & 0.663 \\
0.15 & 0.476 & 0.665 & \textbf{0.714} & 0.667 \\
0.20 & 0.486 & 0.654 & \textbf{0.737} & 0.673 \\
\midrule
% \multicolumn{5}{c}{\textbf{Qwen2.5-72B + PTB}} \\ 
% \hline
% 0.02 & 0.526 & 0.563 & 0.614 & 0.603 \\
% 0.10 & 0.527 & 0.570 & 0.649 & 0.620 \\
% 0.15 & 0.528 & 0.564 & 0.666 & 0.632 \\
% 0.20 & 0.533 & 0.566 & 0.668 & 0.635 \\
% \midrule
\multicolumn{5}{c}{\textbf{Qwen2.5-72B + Wikitext-2}} \\ 
\hline
0.02 & 0.482 & \textbf{0.648} & 0.606 & 0.603 \\
0.10 & 0.489 & 0.651 & \textbf{0.769} & 0.682 \\
0.15 & 0.487 & 0.654 & \textbf{0.779} & 0.682 \\
0.20 & 0.492 & 0.656 & \textbf{0.811} & 0.710 \\
\hline
\end{tabular}
\vspace{-0.1in}
\end{table}

\begin{table}[t]
\centering
\caption{COH comparison between the proposed method and baselines with soft prompt trained on the public C4 dataset.}
\label{tab:ours_vs_inferdpt_c4}
\scriptsize
\setlength{\tabcolsep}{4pt}
\renewcommand{\arraystretch}{1.15}
\begin{tabular}{l|c|cc|c|cc}
\hline
\multirow{2}{*}[-0.5ex]{\textbf{ASR}} &
\multirow{2}{*}[-0.5ex]{\textbf{RANTEXT}} &
\multicolumn{2}{c|}{\textbf{InferDPT}} &
\multirow{2}{*}[-0.5ex]{\textbf{DEL}} &
\multicolumn{2}{c}{\textbf{DEL + Local Model}} \\
\cline{3-4} \cline{6-7}
 % & \makecell{w/o local \\ model}
 &
 & \makecell{with local \\ DialoGPT}
 & \makecell{with local \\ OPT-1.3B}
 &
 & \textbf{DialoGPT} & \textbf{OPT-1.3B} \\
\hline
\multicolumn{7}{c}{\textbf{Llama3 + PTB}} \\
\hline
0.02 & 0.517 & 0.581 & 0.565 & 0.523 & \textbf{0.589} & 0.542 \\
0.10 & 0.520 & 0.585 & 0.567 & 0.554 & \textbf{0.592} & 0.574 \\
0.15 & 0.522 & 0.584 & 0.566 & 0.587 & \textbf{0.605} & 0.591 \\
0.20 & 0.526 & 0.583 & 0.561 & \textbf{0.618} & 0.593 & 0.595 \\
\hline
\multicolumn{7}{c}{\textbf{Llama3 + Wikitext-2}} \\
\hline
0.02 & 0.477 & 0.515 & \textbf{0.665} & 0.514 & 0.520 & 0.594 \\
0.10 & 0.492 & 0.517 & 0.639 & \textbf{0.642} & 0.519 & 0.642 \\
0.15 & 0.493 & 0.518 & 0.648 & \textbf{0.735} & 0.523 & 0.689 \\
0.20 & 0.492 & 0.519 & 0.658 & \textbf{0.756} & 0.521 & 0.701 \\
\hline
\multicolumn{7}{c}{\textbf{Qwen2.5 + PTB}} \\
\hline
0.02 & 0.517 & 0.570 & 0.558 & 0.537 & \textbf{0.589} & 0.548 \\
0.10 & 0.521 & 0.575 & 0.568 & \textbf{0.611} & 0.591 & 0.601 \\
0.15 & 0.526 & 0.574 & 0.561 & \textbf{0.625} & 0.594 & 0.596 \\
0.20 & 0.525 & 0.571 & 0.565 & \textbf{0.625} & 0.588 & 0.609 \\
\hline
\multicolumn{7}{c}{\textbf{Qwen2.5 + Wikitext-2}} \\
\hline
0.02 & 0.486 & 0.515 & \textbf{0.649} & 0.602 & 0.520 & 0.610 \\
0.10 & 0.489 & 0.514 & 0.638 & \textbf{0.746} & 0.522 & 0.697 \\
0.15 & 0.493 & 0.516 & 0.640 & \textbf{0.763} & 0.522 & 0.676 \\
0.20 & 0.493 & 0.516 & 0.642 & \textbf{0.768} & 0.521 & 0.695 \\
\hline
\end{tabular}
\vspace{-0.25in}
\end{table}

% While the proposed mechanism demonstrates a remarkable privacy–utility trade-off in previous experiments, the training splits of the evaluation datasets are not always accessible to the server due to the personalized nature of users’ query styles. 
This section examines the transferability of soft prompt across different datasets and privacy requirement $\mu$.
As shown in Table~\ref{tab:ours_vs_inferdpt_c4}, when the soft prompt is trained on the public C4 dataset and then transferred to other evaluation datasets, the proposed framework still exhibits a better trade-off between privacy and utility than InferDPT. This observation highlights the strong cross-dataset transferability of soft prompt. Furthermore, the results imply better transferability under weaker privacy conditions. For instance, compared with Table \ref{tab:results_transposed}, when using Qwen2.5-7B as the LLM and Wikitext-2 as the evaluation dataset, the coherence at ASR = 0.2 decreases by only 0.032, whereas under stronger privacy (ASR = 0.02), it declines by 0.088. As shown in Table~\ref{tab:soft_mu_transfer}, the soft prompt also exhibits strong transferability across different $\mu$ values.

% \subsection{Ablation Study}
% In this section, we present ablation studies to examine the impacts of the key hyperparameters. Unless otherwise stated, experiments are conducted on WikiText-2 and Llama3-8B. 

\begin{table}[t]
\centering
\caption{Transfer performance (PPL) of soft prompt trained with $\mu=52$ when transferred to other $\mu$ values on LLaMA3-8B with Wikitext2.}
\label{tab:soft_mu_transfer}
\scriptsize
\setlength{\tabcolsep}{3.5pt}  % 缩小列间距
\renewcommand{\arraystretch}{1.1}  % 稍微紧凑
\begin{tabular}{c|cccccc}
\hline
\textbf{$\mu$} & \textbf{40} & \textbf{44} & \textbf{48} & \textbf{52} & \textbf{56} & \textbf{60}\\
\hline
\makecell{soft prompt \\ trained at target $\mu$} & 32.09 & 23.91 & 19.58 & 16.76 & 14.78 & 13.30\\
\hline
\makecell{soft prompt \\ trained at $\mu=52$} & 35.91 & 25.19 & 19.83 & 16.76 & 14.82 & 13.54\\
\hline
No soft prompt & 4279.84 & 3749.72 & 3215.54 & 2738.83 & 2337.73 & 1999.32\\
\hline
\end{tabular}
\vspace{-0.1in}
\end{table}

\begin{table}[t]
\centering
\caption{Performance (PPL) comparison of the latent dimension $d$ on WikiText-2 and Llama3-8B.}
\label{tab:d_values}
\scriptsize
\setlength{\tabcolsep}{5pt}
\renewcommand{\arraystretch}{1.2}
\begin{tabular}{c|cccccc}
\hline
\textbf{$\mu$} & \textbf{20} & \textbf{28} & \textbf{36} & \textbf{44} & \textbf{52} & \textbf{60} \\
\hline
$d=32~(b/128)$ & \textbf{116.01} & \textbf{58.89} & 40.80 & 30.63 & 26.68 & 24.08 \\
\hline
$d=128~(b/32)$ & 195.75 & 88.62 & \textbf{39.42} & \textbf{23.90} & \textbf{16.76} & \textbf{13.30} \\
\hline
$d=512~(b/8)$ & 327.84 & 235.39 & 157.68 & 99.78 & 58.50 & 37.75 \\
\hline
$d=4096~(b)$ & 426.23 & 420.61 & 407.47 & 337.42 & 303.89 & 262.11 \\
\hline
\end{tabular}
\vspace{-0.2in}
\end{table}

% \textbf{Impact of latent dimension $d$:}
\subsection{Impact of Latent Dimension}
Table~\ref{tab:d_values} presents the performance across different latent dimensions $d$.
Under tighter privacy regimes (e.g., $\mu\in\{20,28\}$), a smaller $d$ requires a smaller amount of perturbation to satisfy a fixed differential privacy level, and thus yields better utility despite the reduced semantic capacity.
As the privacy requirement becomes less stringent (i.e., a larger $\mu$), higher-dimensional embeddings (e.g., $d=128$) may capture more semantic information.
However, excessively large dimensions (e.g., $d=512$) still degrade performance due to the large amount of perturbation required to satisfy the same differential privacy level. These results align with the analysis in Section~\ref{privacy_section}, and $d=b/32$ is adopted as a balanced choice.

\section{Conclusion}
We presented DEL, a differentially private and utility-preserving framework for LLM inference. By integrating an encoder–decoder-based dimensionality reduction module with a stochastic $n$-bit quantization mechanism, and further restoring utility through server-side soft prompt, DEL achieved a favorable trade-off between privacy and utility. Experimental results on both text generation and understanding tasks validated the effectiveness of the proposed method.

\section*{Impact Statement}
This paper presents work whose goal is to advance the field of Machine Learning. It addresses privacy risks that arise when user's raw prompts are transmitted to widely deployed cloud-based LLM inference services. A differentially private and communication-efficient mechanism is introduced to protect the original user input, alongside a soft prompt-based utility-preserving framework that maintains inference quality under perturbed inputs. The proposed method achieves a stronger privacy–utility trade-off than prior state-of-the-art methods, while avoiding reliance on local post-processing models on resource-constrained user devices. These advantages may help make private LLM inference more practical and scalable for real-world applications.

% Authors are \textbf{required} to include a statement of the potential 
% broader impact of their work, including its ethical aspects and future 
% societal consequences. This statement should be in an unnumbered 
% section at the end of the paper (co-located with Acknowledgements -- 
% the two may appear in either order, but both must be before References), 
% and does not count toward the paper page limit. In many cases, where 
% the ethical impacts and expected societal implications are those that 
% are well established when advancing the field of Machine Learning, 
% substantial discussion is not required, and a simple statement such 
% as the following will suffice:

% ``This paper presents work whose goal is to advance the field of 
% Machine Learning. There are many potential societal consequences 
% of our work, none which we feel must be specifically highlighted here.''

% The above statement can be used verbatim in such cases, but we 
% encourage authors to think about whether there is content which does 
% warrant further discussion, as this statement will be apparent if the 
% paper is later flagged for ethics review.


\bibliography{ref_llm}

@InProceedings{wu2025cape,
  title     = {Cape: Context-Aware Prompt Perturbation Mechanism with Differential Privacy},
  author    = {Wu, Haoqi and Dai, Wei and Li, Wang and Yan, Qiang},
  booktitle = {Proceedings of the 42nd International Conference on Machine Learning},
  pages     = {67184--67201},
  year      = {2025},
  volume    = {267},
  publisher = {PMLR},
}

@inproceedings{brown2020language,
 author = {Brown, Tom and Mann, Benjamin and Ryder, Nick and Subbiah, Melanie and Kaplan, Jared D and Dhariwal, Prafulla and Neelakantan, Arvind and Shyam, Pranav and Sastry, Girish and Askell, Amanda and Agarwal, Sandhini and Herbert-Voss, Ariel and Krueger, Gretchen and Henighan, Tom and Child, Rewon and Ramesh, Aditya and Ziegler, Daniel and Wu, Jeffrey and Winter, Clemens and Hesse, Chris and Chen, Mark and Sigler, Eric and Litwin, Mateusz and Gray, Scott and Chess, Benjamin and Clark, Jack and Berner, Christopher and McCandlish, Sam and Radford, Alec and Sutskever, Ilya and Amodei, Dario},
 booktitle = {Advances in Neural Information Processing Systems},
 pages = {1877--1901},
 title = {Language Models are Few-Shot Learners},
 volume = {33},
 year = {2020}
}

@inproceedings{duan2023privacy,
  title     = {Flocks of Stochastic Parrots: Differentially Private Prompt Learning for Large Language Models},
  author    = {Duan, Haonan and Dziedzic, Adam and Papernot, Nicolas and Boenisch, Franziska},
  booktitle = {Advances in Neural Information Processing Systems (NeurIPS)},
  year      = {2023},
}

@inproceedings{yue2021differential,
  title     = {Differential Privacy for Text Analytics via Natural Text Sanitization},
  author    = {Yue, Xiang and Du, Minxin and Wang, Tianhao and Li, Yaliang and Sun, Huan and Chow, Sherman S. M.},
  booktitle = {Findings of the Association for Computational Linguistics: ACL-IJCNLP 2021},
  month     = aug,
  year      = {2021},
  address   = {Online},
  publisher = {Association for Computational Linguistics},
  doi       = {10.18653/v1/2021.findings-acl.337},
  pages     = {3853--3866}
}

@inproceedings{chen2023customized,
  title     = {A Customized Text Sanitization Mechanism with Differential Privacy},
  author    = {Chen, Sai and Mo, Fengran and Wang, Yanhao and Chen, Cen and Nie, Jian-Yun and Wang, Chengyu and Cui, Jamie},
  booktitle = {Findings of the Association for Computational Linguistics: ACL 2023},
  month     = jul,
  year      = {2023},
  address   = {Toronto, Canada},
  publisher = {Association for Computational Linguistics},
  doi       = {10.18653/v1/2023.findings-acl.355},
  pages     = {5747--5758}
}

@inproceedings{tang2024privacy_preserving_icl,
  title     = {Privacy-Preserving In-Context Learning with Differentially Private Few-Shot Generation},
  author    = {Tang, Xinyu and Shin, Richard and Inan, Huseyin A. and Manoel, Andre and Mireshghallah, Fatemehsadat and Lin, Zinan and Gopi, Sivakanth and Kulkarni, Janardhan and Sim, Robert},
  booktitle = {International Conference on Learning Representations (ICLR)},
  year      = {2024},
}

@inproceedings{hermann2015teaching,
  title={Teaching machines to read and comprehend},
  author={Hermann, Karl Moritz and Ko{\v{c}}isk{\'y}, Tom{\'a}{\v{s}} and Grefenstette, Edward and Espeholt, Lasse and Kay, Will and Suleyman, Mustafa and Blunsom, Phil},
  booktitle={Advances in Neural Information Processing Systems},
  volume={28},
  pages={1693--1701},
  year={2015},
}

@article{zhang2022opt,
  title={OPT: Open Pre-trained Transformer Language Models},
  author={Zhang, Susan and Roller, Stephen and Goyal, Naman and Artetxe, Mikel and Chen, Moya and Chen, Shuohui and Dewan, Christopher and Diab, Mona and Li, Xian and Lin, Xi Victoria and Mihaylov, Todor and Ott, Myle and Shleifer, Sam and Shuster, Kurt and Simig, Daniel and Koura, Punit Singh and Sridhar, Anjali and Wang, Tianlu and Zettlemoyer, Luke},
  journal={arXiv preprint arXiv:2205.01068},
  year={2022},
}

@article{liu2023llms,
  title={LLMs Can Understand Encrypted Prompt: Towards Privacy-Computing Friendly Transformers},
  author={Liu, Xuanqi and Liu, Zhuotao},
  journal={arXiv preprint arXiv:2305.18396},
  year={2023},
}

@article{hou2023ciphergpt,
  title={Ciphergpt: Secure two-party gpt inference},
  author={Hou, Xiaoyang and Liu, Jian and Li, Jingyu and Li, Yuhan and Lu, Wen-jie and Hong, Cheng and Ren, Kui},
  journal={Cryptology ePrint Archive},
  year={2023}
}

@inproceedings{mai2024split,
  title={Split-and-denoise: Protect large language model inference with local differential privacy},
  author={Mai, Peihua and Yan, Ran and Huang, Zhe and Yang, Youjia and Pang, Yan},
  booktitle={Proceedings of the 41st International Conference on Machine Learning},
  pages={33440--33460},
  year={2024},
  volume={235},
  series={Proceedings of Machine Learning Research},
  publisher={PMLR},
}

@article{tong2025inferdpt,
  title={InferDPT: Privacy-preserving Inference for Black-box Large Language Models},
  author={Tong, Meng and Chen, Kejiang and Zhang, Jia and others},
  journal={IEEE Transactions on Dependable and Secure Computing},
  year={2025},
  publisher={IEEE},
}

@inproceedings{Hoory2021LearningAE,
  title     = {Learning and Evaluating a Differentially Private Pre-trained Language Model},
  author    = {Hoory, Shlomo and Feder, Amir and Tendler, Avichai and Erell, Sofia and Peled-Cohen, Alon and Laish, Itay and Nakhost, Hootan and Stemmer, Uri and Benjamini, Ayelet and Hassidim, Avinatan and Matias, Yossi},
  booktitle = {Findings of the Association for Computational Linguistics: EMNLP 2021},
  year      = {2021},
  month     = nov,
  pages     = {1178--1189},
  publisher = {Association for Computational Linguistics},
  doi       = {10.18653/v1/2021.findings-emnlp.102}
}

@inproceedings{yin-habernal-2022-privacy,
    title = "Privacy-Preserving Models for Legal Natural Language Processing",
    author = "Yin, Ying  and
      Habernal, Ivan",
    booktitle = "Proceedings of the Natural Legal Language Processing Workshop 2022",
    month = dec,
    year = "2022",
    doi = "10.18653/v1/2022.nllp-1.14",
    pages = "172--183",
}

@InProceedings{Ganesh2023WhyIP,
  title = 	 {Why Is Public Pretraining Necessary for Private Model Training?},
  author =       {Ganesh, Arun and Haghifam, Mahdi and Nasr, Milad and Oh, Sewoong and Steinke, Thomas and Thakkar, Om and Guha Thakurta, Abhradeep and Wang, Lun},
  booktitle = 	 {Proceedings of the 40th International Conference on Machine Learning},
  pages = 	 {10611--10627},
  year = 	 {2023},
  volume = 	 {202},
  series = 	 {Proceedings of Machine Learning Research},
  month = 	 {23--29 Jul},
  publisher =    {PMLR},
}

@inproceedings{Li2021LargeLM,
title={Large Language Models Can Be Strong Differentially Private Learners},
author={Xuechen Li and Florian Tramer and Percy Liang and Tatsunori Hashimoto},
booktitle={International Conference on Learning Representations},
year={2022},
}

@InProceedings{Bu2022DifferentiallyPO,
  title = 	 {Differentially Private Optimization on Large Model at Small Cost},
  author =       {Bu, Zhiqi and Wang, Yu-Xiang and Zha, Sheng and Karypis, George},
  booktitle = 	 {Proceedings of the 40th International Conference on Machine Learning},
  pages = 	 {3192--3218},
  year = 	 {2023},
  volume = 	 {202},
  series = 	 {Proceedings of Machine Learning Research},
  month = 	 {23--29 Jul},
  publisher =    {PMLR},
}

@article{kurakin2023harnessing,
  title={Harnessing large-language models to generate private synthetic text},
  author={Kurakin, Alexey and Ponomareva, Natalia and Syed, Umar and MacDermed, Liam and Terzis, Andreas},
  journal={arXiv preprint arXiv:2306.01684},
  year={2023}
}

@inproceedings{Yu2021DifferentiallyPF,
title={Differentially Private Fine-tuning of Language Models},
author={Da Yu and Saurabh Naik and Arturs Backurs and Sivakanth Gopi and Huseyin A Inan and Gautam Kamath and Janardhan Kulkarni and Yin Tat Lee and Andre Manoel and Lukas Wutschitz and Sergey Yekhanin and Huishuai Zhang},
booktitle={International Conference on Learning Representations},
year={2022},
}

@article{Li2023PrivacyPreservingPT,
  title={Privacy-Preserving Prompt Tuning for Large Language Model Services},
  author={Yansong Li and Zhixing Tan and Yang Liu},
  journal={arXiv preprint arXiv:2305.06212},
  year={2023}
}

@article{Chen2023HideAS,
  title={Hide and Seek (HaS): A Lightweight Framework for Prompt Privacy Protection},
  author={Yu Chen and Tingxin Li and Huiming Liu and Yang Yu},
  journal={arXiv preprint arXiv:2309.03057},
  year={2023},
}

@article{li2021prefix,
  title={Prefix-tuning: Optimizing continuous prompts for generation},
  author={Li, Xiang Lisa and Liang, Percy},
  journal={arXiv preprint arXiv:2101.00190},
  year={2021}
}

@article{shin2020autoprompt,
  title={AutoPrompt: Eliciting Knowledge from Language Models with Automatically Generated Prompts},
  author={Shin, Taylor and Razeghi, Yasaman and Logan IV, Robert L. and Wallace, Eric and Singh, Sameer},
  journal={arXiv preprint arXiv:2010.15980},
  year={2020}
}

@article{lester2021power,
  title={The Power of Scale for Parameter-Efficient Prompt Tuning},
  author={Lester, Brian and Al-Rfou, Rami and Constant, Noah},
  journal={arXiv preprint arXiv:2104.08691},
  year={2021}
}

@article{wu2023infoprompt,
  title={Infoprompt: Information-theoretic soft prompt tuning for natural language understanding},
  author={Wu, Junda and Yu, Tong and Wang, Rui and Song, Zhao and Zhang, Ruiyi and Zhao, Handong and Lu, Chaochao and Li, Shuai and Henao, Ricardo},
  journal={Advances in Neural Information Processing Systems},
  volume={36},
  pages={61060--61084},
  year={2023}
}

@inproceedings{liu2024can,
  title={Can we soft prompt LLMs for graph learning tasks?},
  author={Liu, Zheyuan and He, Xiaoxin and Tian, Yijun and Chawla, Nitesh V},
  booktitle={Companion Proceedings of the ACM Web Conference 2024},
  pages={481--484},
  year={2024}
}

@inproceedings{xu2024soft,
  title={Soft prompt recovers compressed llms, transferably},
  author={Xu, Zhaozhuo and Liu, Zirui and Chen, Beidi and Zhong, Shaochen and Tang, Yuxin and Wang, Jue and Zhou, Kaixiong and Hu, Xia and Shrivastava, Anshumali},
  booktitle={Forty-first International Conference on Machine Learning},
  year={2024}
}

@article{su2021transferability,
  title={On transferability of prompt tuning for natural language processing},
  author={Su, Yusheng and Wang, Xiaozhi and Qin, Yujia and Chan, Chi-Min and Lin, Yankai and Wang, Huadong and Wen, Kaiyue and Liu, Zhiyuan and Li, Peng and Li, Juanzi and others},
  journal={arXiv preprint arXiv:2111.06719},
  year={2021}
}

@article{liu2024mitigating,
  title        = {Mitigating Privacy Risks in LLM Embeddings from Embedding Inversion},
  author       = {Liu, Tiantian and Yao, Hongwei and Wu, Tong and Qin, Zhan and Lin, Feng and Ren, Kui and Chen, Chun},
  journal      = {arXiv preprint arXiv:2411.05034},
  year         = {2024},
}

@article{chen2024text,
  title        = {Text Embedding Inversion Security for Multilingual Language Models},
  author       = {Chen, Yiyi and Lent, Heather and Bjerva, Johannes},
  journal      = {arXiv preprint arXiv:2401.12192},
  year         = {2024},
}

@article{wan2024information,
  title        = {Information Leakage from Embedding in Large Language Models},
  author       = {Wan, Zhipeng and Cheng, Anda and Wang, Yinggui and Wang, Lei},
  journal      = {arXiv preprint arXiv:2405.11916},
  year         = {2024},
}

@article{dong2022gaussian,
  title={Gaussian Differential Privacy},
  author={Dong, Jinshuo and Roth, Aaron and Su, Weijie J.},
  journal={Journal of the Royal Statistical Society: Series B (Statistical Methodology)},
  volume={84},
  number={1},
  pages={3--37},
  year={2022},
  publisher={Wiley},
  doi={10.1111/rssb.12454}
}

@inproceedings{abadi2016deep,
  title={Deep learning with differential privacy},
  author={Abadi, Martin and Chu, Andy and Goodfellow, Ian and McMahan, H Brendan and Mironov, Ilya and Talwar, Kunal and Zhang, Li},
  booktitle={Proceedings of the 2016 ACM SIGSAC Conference on Computer and Communications Security},
  pages={308--318},
  year={2016}
}

@inproceedings{chen2023privacy,
  title     = {Privacy Amplification via Compression: Achieving the Optimal Privacy-Accuracy-Communication Trade-off in Distributed Mean Estimation},
  author    = {Chen, Wei-Ning and Song, Dan and Özgür, Ayfer and Kairouz, Peter},
  booktitle = {Advances in Neural Information Processing Systems},
  year      = {2023},
}

@book{lehmann2005testing,
  title={Testing statistical hypotheses},
  author={Lehmann, Erich Leo and Romano, Joseph P and Casella, George},
  volume={3},
  year={2005},
  publisher={Springer}
}

@article{Raffel2020T5,
  author  = {Colin Raffel and Noam Shazeer and Adam Roberts and Katherine Lee and
             Sharan Narang and Michael Matena and Yanqi Zhou and Wei Li and Peter J. Liu},
  title   = {Exploring the Limits of Transfer Learning with a Unified Text-to-Text Transformer},
  journal = {Journal of Machine Learning Research},
  year    = {2020},
  volume  = {21},
  number  = {140},
  pages   = {1--67},
}

@inproceedings{Merity2017PointerSentinel,
  author    = {Stephen Merity and Caiming Xiong and James Bradbury and Richard Socher},
  title     = {Pointer Sentinel Mixture Models},
  booktitle = {Proceedings of the 5th International Conference on Learning Representations (ICLR)},
  year      = {2017},
  month     = {April},
}

@inproceedings{marcus-etal-1994-penn,
  title     = {The {Penn} {Treebank}: Annotating Predicate Argument Structure},
  author    = {Mitchell Marcus and Grace Kim and Mary Ann Marcinkiewicz and Robert MacIntyre and Ann Bies and Mark Ferguson and Karen Katz and Britta Schasberger},
  booktitle = {Human Language Technology: Proceedings of a Workshop held at Plainsboro, New Jersey, March 8-11, 1994},
  year      = {1994},
}

@misc{Touvron2023,
  title        = {Llama 2: Open Foundation and Fine-Tuned Chat Models},
  author       = {Hugo Touvron and Louis Martin and Kevin R. Stone and Peter Albert
                  and Amjad Almahairi and Yasmine Babaei and Nikolay Bashlykov
                  and Soumya Batra and Prajjwal Bhargava and Shruti Bhosale
                  and Dan Bikel and Lukas Blecher and Cristian Cantón Ferrer
                  and Moya Chen and Guillem Cucurull and David Esiobu and Jude Fernandes
                  and Jeremy Fu and Wenyin Fu and Brian Fuller and Cynthia Gao
                  and Vedanuj Goswami and Naman Goyal and Anthony Hartshorn
                  and Saghar Hosseini and Rui Hou and Hakan Inan and Marcin Kardas
                  and Viktor Kerkez and Madian Khabsa and Isabel M. Kloumann
                  and Artem Korenev and Punit Singh Koura and Marie-Anne Lachaux
                  and Thibaut Lavril and Jenya Lee and Diana Liskovich and Yinghai Lu
                  and Yuning Mao and Xavier Martinet and Todor Mihaylov and Pushkar Mishra
                  and Igor Molybog and Yixin Nie and Andrew Poulton and Jeremy Reizenstein
                  and Rashi Rungta and Kalyan Saladi and Alan Schelten and Ruan Silva
                  and Eric Michael Smith and Ranjan Subramanian and Xiaoqing Ellen Tan
                  and Binh Tang and Ross Taylor and Adina Williams
                  and Jian Xiang Kuan and Puxin Xu and Zheng Yan and Iliyan Zarov
                  and Yuchen Zhang and Angela Fan and Melanie Kambadur and Sharan Narang
                  and Aurélien Rodriguez and Robert Stojnic and Sergey Edunov and Thomas Scialom},
  year         = {2023},
  month        = {Jul},
  day          = {18},
}

@misc{qwen2025qwen25,
  title        = {Qwen2.5 Technical Report},
  author       = {Qwen and An Yang and Baosong Yang and Beichen Zhang and Binyuan Hui and Bo Zheng and Bowen Yu and Chengyuan Li and Dayiheng Liu and Fei Huang and Haoran Wei and Huan Lin and Jian Yang and Jianhong Tu and Jianwei Zhang and Jianxin Yang and Jiaxi Yang and Jingren Zhou and Junyang Lin and Kai Dang and Keming Lu and Keqin Bao and Kexin Yang and Le Yu and Mei Li and Mingfeng Xue and Pei Zhang and Qin Zhu and Rui Men and Runji Lin and Tianhao Li and Tianyi Tang and Tingyu Xia and Xingzhang Ren and Xuancheng Ren and Yang Fan and Yang Su and Yichang Zhang and Yu Wan and Yuqiong Liu and Zeyu Cui and Zhenru Zhang and Zihan Qiu},
  year         = {2025},
  month        = {January},
  doi          = {10.48550/arXiv.2412.15115}
}

@inproceedings{gao-etal-2021-simcse,
  title     = {SimCSE: Simple Contrastive Learning of Sentence Embeddings},
  author    = {Tianyu Gao and Xingcheng Yao and Danqi Chen},
  booktitle = {Proceedings of the 2021 Conference on Empirical Methods in Natural Language Processing},
  month     = nov,
  year      = {2021},
  address   = {Online and Punta Cana, Dominican Republic},
  publisher = {Association for Computational Linguistics},
  pages     = {6894--6910},
  doi       = {10.18653/v1/2021.emnlp-main.552},
}

@misc{quora2017qqp,
  author  = {Iyer, Shankar and Dandekar, Nikhil and Csernai, Korn{\'e}l},
  title   = {First Quora Dataset Release: Question Pairs},
  year    = {2017},
}

@inproceedings{dolan-brockett-2005-automatically,
  title     = {Automatically Constructing a Corpus of Sentential Paraphrases},
  author    = {William B. Dolan and Chris Brockett},
  booktitle = {Proceedings of the Third International Workshop on Paraphrasing (IWP2005)},
  year      = {2005},
}

@inproceedings{devlin2019bert,
  title     = {BERT: Pre-training of Deep Bidirectional Transformers for Language Understanding},
  author    = {Devlin, Jacob and Chang, Ming-Wei and Lee, Kenton and Toutanova, Kristina},
  booktitle = {Proceedings of the 2019 Conference of the North American Chapter of the Association for Computational Linguistics: Human Language Technologies, Volume 1 (Long and Short Papers)},
  year      = {2019},
  pages     = {4171--4186},
  publisher = {Association for Computational Linguistics},
  doi       = {10.18653/v1/N19-1423},
}

@inproceedings{zhang-etal-2020-dialogpt,
  title     = {DIALOGPT : Large-Scale Generative Pre-training for Conversational Response Generation},
  author    = {Zhang, Yizhe and Sun, Siqi and Galley, Michel and Chen, Yen-Chun and Brockett, Chris and Gao, Xiang and Gao, Jianfeng and Liu, Jingjing and Dolan, Bill},
  booktitle = {Proceedings of the 58th Annual Meeting of the Association for Computational Linguistics: System Demonstrations},
  month     = jul,
  year      = {2020},
  address   = {Online},
  publisher = {Association for Computational Linguistics},
  doi       = {10.18653/v1/2020.acl-demos.30},
  pages     = {270--278}
}

@article{LLMSurvey,
  title   = {A Survey of Large Language Models},
  author  = {Zhao, Wayne Xin and Zhou, Kun and Li, Junyi and Tang, Tianyi and Wang, Xiaolei and Hou, Yupeng and Min, Yingqian and Zhang, Beichen and Zhang, Junjie and Dong, Zican and Du, Yifan and Yang, Chen and Chen, Yushuo and Chen, Zhipeng and Jiang, Jinhao and Ren, Ruiyang and Li, Yifan and Tang, Xinyu and Liu, Zikang and Liu, Peiyu and Nie, Jian-Yun and Wen, Ji-Rong},
  journal = {arXiv preprint arXiv:2303.18223},
  year    = {2023},
}

@article{frantar2022gptq,
  title   = {GPTQ: Accurate Post-Training Quantization for Generative Pre-trained Transformers},
  author  = {Frantar, Elias and Ashkboos, Saleh and Hoefler, Torsten and Alistarh, Dan},
  journal = {arXiv preprint arXiv:2210.17323},
  year    = {2022},
}

@inproceedings{liu2021highdimdp,
  title={Differential Privacy and Robust Statistics in High Dimensions},
  author={Liu, Jing and Duchi, John C. and Jordan, Michael I.},
  booktitle={Proceedings of the AAAI Privacy-Preserving Artificial Intelligence Workshop (PPAI)},
  year={2021}
}

@article{liu2023csur,
  title   = {Pre-train, Prompt, and Predict: A Systematic Survey of Prompting Methods in Natural Language Processing},
  author  = {Liu, Pengfei and Yuan, Weizhe and Fu, Jinlan and Jiang, Zhengbao and Hayashi, Hiroaki and Neubig, Graham},
  journal = {ACM Computing Surveys},
  volume  = {55},
  number  = {9},
  pages   = {1--35},
  year    = {2023},
  doi     = {10.1145/3560815},
}

@inproceedings{Dai2024DeepSeekMoE,
  author    = {Damai Dai and Chengqi Deng and Chenggang Zhao and R.\,X.\ Xu and Huazuo Gao and Deli Chen and Jiashi Li and Wangding Zeng and Xingkai Yu and Y.\,Wu and Zhenda Xie and Y.\,K.\ Li and Panpan Huang and Fuli Luo and Chong Ruan and Zhifang Sui and Wenfeng Liang},
  title     = {DeepSeekMoE: Towards Ultimate Expert Specialization in Mixture-of-Experts Language Models},
  booktitle = {Proceedings of the 62nd Annual Meeting of the Association for Computational Linguistics (ACL 2024) – Volume 1: Long Papers},
  pages     = {1280--1297},
  year      = {2024},
  address   = {Bangkok, Thailand},
  publisher = {Association for Computational Linguistics},
  doi       = {10.18653/v1/2024.acl-long.70},
}

@article{Chen2025PanguEmbedded,
  author    = {Hanting Chen and Yasheng Wang and Kai Han and Dong Li and Lin Li and Zhenni Bi and Jinpeng Li and Haoyu Wang and Fei Mi and Mingjian Zhu and Bin Wang and Kaikai Song and Yifei Fu and Xu He and Yu Luo and Chong Zhu and Quan He and Xueyu Wu and Wei He and Hailin Hu and Yehui Tang and Dacheng Tao and Xinghao Chen and Yunhe Wang},
  title     = {Pangu Embedded: An Efficient Dual-system LLM Reasoner with Metacognition},
  journal   = {arXiv preprint arXiv:2505.22375},
  year      = {2025},
}

@inproceedings{Alistarh2017QSGD,
  title     = {{QSGD}: Communication-Efficient SGD via Gradient Quantization and Encoding},
  author    = {Alistarh, Dan and Grubic, Demjan and Li, Jerry and Tomioka, Ryota and Vojnovic, Milan},
  booktitle = {Advances in Neural Information Processing Systems},
  volume    = {30},
  pages     = {1709--1720},
  year      = {2017},
}

@inproceedings{He2019MIACollabInference,
  author    = {Zecheng He and Tianwei Zhang and Ruby B. Lee},
  title     = {Model Inversion Attacks Against Collaborative Inference},
  booktitle = {Proceedings of the 35th Annual Computer Security Applications Conference (ACSAC '19)},
  series    = {ACM International Conference Proceeding Series},
  year      = {2019},
  pages     = {148--162},
  publisher = {Association for Computing Machinery},
  address   = {New York, NY, USA},
  location  = {San Juan, PR, USA},
  isbn      = {978-1-4503-7628-0},
  doi       = {10.1145/3359789.3359824},
}

@inproceedings{wu24ditto,
  title     = {Ditto: Quantization-aware Secure Inference of Transformers upon {MPC}},
  author    = {Wu, Haoqi and Fang, Wenjing and Zheng, Yancheng and Ma, Junming and Tan, Jin and Wang, Lei},
  booktitle = {Proceedings of the 41st International Conference on Machine Learning},
  pages     = {53346--53365},
  year      = {2024},
  volume    = {235},
  series    = {Proceedings of Machine Learning Research},
  month     = {21--27 Jul},
  publisher = {PMLR},
}

@inproceedings{zeng2025privacyrestore,
  title     = {PrivacyRestore: Privacy-Preserving Inference in Large Language Models via Privacy Removal and Restoration},
  author    = {Zeng, Ziqian and Wang, Jianwei and Yang, Junyao and Lu, Zhengdong and Li, Haoran and Zhuang, Huiping and Chen, Cen},
  booktitle = {Proceedings of the 63rd Annual Meeting of the Association for Computational Linguistics (Volume 1: Long Papers)},
  pages     = {10821--10855},
  year      = {2025},
  month     = jul,
  address   = {Vienna, Austria},
  publisher = {Association for Computational Linguistics},
  doi       = {10.18653/v1/2025.acl-long.532},
  isbn      = {979-8-89176-251-0}
}

@inproceedings{zhou2022textfusion,
  title     = {TextFusion: Privacy-Preserving Pre-trained Model Inference via Token Fusion},
  author    = {Zhou, Xin and Lu, Jinzhu and Gui, Tao and Ma, Ruotian and Fei, Zichu and Wang, Yuran and Ding, Yong and Cheung, Yibo and Zhang, Qi and Huang, Xuanjing},
  booktitle = {Proceedings of the 2022 Conference on Empirical Methods in Natural Language Processing},
  pages     = {8360--8371},
  year      = {2022},
  month     = dec,
  address   = {Abu Dhabi, United Arab Emirates},
  publisher = {Association for Computational Linguistics},
  doi       = {10.18653/v1/2022.emnlp-main.572}
}

@inproceedings{du2023dpforward,
  title     = {{DP-Forward}: Fine-tuning and Inference on Language Models with Differential Privacy in Forward Pass},
  author    = {Du, Minxin and Yue, Xiang and Chow, Sherman S. M. and Wang, Tianhao and Huang, Chenyu and Sun, Huan},
  booktitle = {Proceedings of the 2023 ACM SIGSAC Conference on Computer and Communications Security (CCS~'23)},
  pages     = {2665--2679},
  year      = {2023},
  address   = {Copenhagen, Denmark},
  publisher = {Association for Computing Machinery},
  doi       = {10.1145/3576915.3616592},
}
\bibliographystyle{icml2026}

%%%%%%%%%%%%%%%%%%%%%%%%%%%%%%%%%%%%%%%%%%%%%%%%%%%%%%%%%%%%%%%%%%%%%%%%%%%%%%%
%%%%%%%%%%%%%%%%%%%%%%%%%%%%%%%%%%%%%%%%%%%%%%%%%%%%%%%%%%%%%%%%%%%%%%%%%%%%%%%
% APPENDIX
%%%%%%%%%%%%%%%%%%%%%%%%%%%%%%%%%%%%%%%%%%%%%%%%%%%%%%%%%%%%%%%%%%%%%%%%%%%%%%%
%%%%%%%%%%%%%%%%%%%%%%%%%%%%%%%%%%%%%%%%%%%%%%%%%%%%%%%%%%%%%%%%%%%%%%%%%%%%%%%
\newpage
\appendix
\onecolumn

\section{Proof of Theorem \ref{theorem_sto}}
\begin{proof}
    First, we will prove the unbiasedness. The expected value of perturbed $v_{i,j}$ is given by:
\[
\mathbb{E}[\mathcal{M}^{\text{sto}}(v_{i,j}; A, n)] = \mathbb{E}\left[\frac{2K - (2^n - 1)}{2^n - 1} \cdot A \right]
= \frac{2A}{2^n - 1} \mathbb{E}[K] - A.
\]

Since $K$ follows a binomial distribution, we have:
\[
\mathbb{E}[K] = u \cdot p(v_{i,j}) = (2^n - 1) \cdot \frac{A + v_{i,j}}{2A}.
\]

Substituting this back into the expression for the expected value:
\[
\mathbb{E}[\mathcal{M}^{\text{sto}}(v_{i,j}; A, n)] = \frac{2A}{2^n - 1} \cdot (2^n - 1) \cdot \frac{A + v_{i,j}}{2A} - A
= v_{i,j}.
\]

$\operatorname{Var}(\mathcal{M}^{\text{sto}}(v_{i,j}; A, n))
=\left(\frac{2A}{u}\right)^2 \operatorname{Var}(K)
=\left(\frac{2A}{u}\right)^2 up(1-p)$.
Using $p(1-p)=\tfrac{A^2-v^2}{4A^2}$, we obtain the closed form
\begin{equation}\label{eq:mse-per-coord}
\operatorname{Var}(\mathcal{M}^{\text{sto}}(v_{i,j}; A, n))
=\frac{A^2-v^2}{2^n-1}.
\end{equation}

Applying the mechanism independently across coordinates for $\bm{v}_i$, the total variance is represented as
\begin{equation}\label{eq:mse-vector}
\operatorname{Var}(\mathcal{M}^{\text{sto}}(\bm{v}_{i}; A, n))
=\frac{dA^2-\|\bm{v}_i\|_2^2}{2^n-1}.
\end{equation}

% \begin{remark}
% Equation~\eqref{eq:mse-per-coord} shows the MSE decreases linearly in the number of quantization steps (u=2^n-1) (i.e., exponentially in (n)), is largest at (v=0), and diminishes as (|v|\to A). For fixed (n), choosing the smallest feasible scaling (A) (typically (A=c)) minimizes the distortion.
% \end{remark}

Next, we will prove the $\mu$-DP guarantee of the stochastic $n$-bit compressor. In fact, it can be viewed as the average of $2^n-1$ independent stochastic binary mechanisms applied coordinate-wise, which can be represented as 
\begin{equation}
\text{sto-binary}(v_{i,j};A)=\begin{cases}
    A,p=\frac{A+v_{i,j}}{2A},\\
    -A,p=\frac{A-v_{i,j}}{2A}.
\end{cases}
\end{equation}
Since the sum of $2^n-1$ such independent outputs can range from $-(2^n-1)$ to $2^n-1$ and it must be an odd multiple of $A$, the resulting average value matches the $2^n$ uniformly spaced quantization levels defined by the set $\mathcal{Q}_n$ in Definition \ref{define_sto}. In other words, once the DP guarantee of the stochastic binary mechanisms is obtained, we can derive that for the proposed stochastic $n$-bit compressor by utilizing the composition theorem \cite{dong2022gaussian}. In the following, we begin with the $f$-DP analysis of each coordinate $v_{i,j}$. Let $Y=\text{sto-binary}(v'_{i,j};A)$ and $X=\text{sto-binary}(v_{i,j};A)$, we first introduce the Neyman-Pearson Lemma \cite{lehmann2005testing}

\begin{lemma}(Neyman-Pearson Lemma \cite{lehmann2005testing})\label{Neyman-Pearson}
Let $P$ and $Q$ be probability distributions on $\Omega$ with densities $p$ and $q$, respectively. For the hypothesis testing problem $H_{0}: P$ vs $H_{1}: Q$, a test $\phi:\Omega \rightarrow [0,1]$ is the most powerful test at level $\alpha$ if and only if there are two constants $h \in [0,+\infty]$ and $\gamma \in [0,1]$ such that $\phi$ has the form 
\begin{equation}\label{rejectionrule}
\phi(x) =
\begin{cases}
\hfill 1, \hfill \text{if $\frac{q(x)}{p(x)} > h$},\\
\hfill \gamma, \hfill \text{if $\frac{q(x)}{p(x)} = h$},\\
\hfill 0, \hfill \text{if $\frac{q(x)}{p(x)} < h$},\\
\end{cases}
\end{equation}
and $\mathbb{E}_{P}[\phi] = \alpha$. The rejection rule suggests that $H_{0}$ is rejected with a probability of $\phi(x)$ given the observation $x$.
\end{lemma}

Given Lemma \ref{Neyman-Pearson}, the problem is then reduced to finding the corresponding $h$ and $\gamma$ such that the type I error rate $\alpha_{\phi} = \alpha$. 

If $v_{i,j}>v'_{i,j}$, when $\alpha\leq P(X=-1)$, we set $h=\frac{P(Y=-1)}{P(X=-1)}$. In this case, $\alpha=\gamma P(X=-1)$. We adjust $\gamma$ such that $\mathbb{E}_{P}[\phi] = \alpha$, which yields $\gamma=\frac{\alpha}{P(X=-1)}$ and
\begin{equation}
    \beta_\phi(\alpha)=1-\mathbb{E}_{Q}[\phi]=1-\gamma P(Y=-1)=1-\frac{P(Y=-1)}{P(X=-1)}\alpha.
\end{equation}
When $P(X=-1)<\alpha\leq 1$, we set $h=\frac{P(Y=1)}{P(X=1)}$, and $\frac{P(Y=k')}{P(X=k')}>h$ for $k'<1$. In this case, $\alpha=\gamma P(X=1)+P(X=-1)$. To satisfy $\mathbb{E}_{P}[\phi] = \alpha$, we derive $\gamma=\frac{\alpha-P(X=-1)}{P(X=1)}$, and
\begin{equation}
    \beta_\phi(\alpha)=1-\mathbb{E}_{Q}[\phi]=1-\gamma P(Y=1)-P(Y=-1)=\frac{P(Y=1)}{P(X=1)}-\frac{P(Y=1)}{P(X=1)}\alpha.
\end{equation}
Combining the above two cases, we obtain
\begin{equation}
    \beta_\phi(\alpha)=\begin{cases}
        1-\frac{P(Y=-1)}{P(X=-1)}\alpha, &\text{for}~~ \alpha\leq P(X=-1),\\
        \frac{P(Y=1)}{P(X=1)}-\frac{P(Y=1)}{P(X=1)}\alpha, &\text{for}~~ P(X=-1)<\alpha\leq 1.
    \end{cases}
\end{equation}

Similarly, if $v_{i,j}<v'_{i,j}$, we can obtain
\begin{equation}
    \beta_\phi(\alpha)=\begin{cases}
        1-\frac{P(Y=1)}{P(X=1)}\alpha, &\text{for}~~ \alpha\leq P(X=1),\\
        \frac{P(Y=-1)}{P(X=-1)}-\frac{P(Y=-1)}{P(X=-1)}\alpha, &\text{for}~~ P(X=1)<\alpha\leq 1.
    \end{cases}
\end{equation}
$\beta_\phi(\alpha)$ attains its minimum when $P(Y=-1)=\frac{A+c}{2A}$ and $P(X=-1)=\frac{A-c}{2A}$ for $v_{i,j}>v'_{i,j}$, and $P(Y=-1)=\frac{A-c}{2A}$ and $P(X=-1)=\frac{A+c}{2A}$ for $v_{i,j}<v'_{i,j}$. As a result, we have
\begin{equation}\label{f_of_scalar}
    f(\alpha)=\begin{cases}
        1-\frac{A+c}{A-c}\alpha, \text{for}~~\alpha\in[0,\frac{A-c}{2A}],\\
        \frac{A-c}{A+c}-\frac{A-c}{A+c}\alpha, \text{for}~~\alpha\in[\frac{A-c}{2A},1].
    \end{cases}
\end{equation}

Next, we extend the scalar case to the vector case. We first introduce the following Lemma

\begin{lemma}[\cite{dong2022gaussian}]\label{cltfdp}
Define the following functions
\begin{equation}
    \text{kl}(f) = -\int_{0}^{1}\log|f'(x)|dx,
\end{equation}
\begin{equation}
    \kappa_{2}(f) = \int_{0}^{1}\log^{2}|f'(x)|dx,
\end{equation}
\begin{equation}
    \kappa_{3}(f) = \int_{0}^{1}|\log|f'(x)||^3dx,
\end{equation}
\begin{equation}
    \bar{\kappa}_{3}(f) = \int_{0}^{1}|\log|f'(x)|+\text{kl}(f)|^3dx,
\end{equation}
let $f_{1},...,f_{n}$ be symmetric trade-off functions such that $\kappa_{3}(f_{i}) < \infty$ for all $1 \leq i \leq n$. Denote
\begin{equation}\nonumber
\mu = \frac{2||\text{kl}||_{1}}{\sqrt{||\kappa_{2}||_{1}-||\text{kl}||_{2}^{2}}}, \text{and~~} \gamma = \frac{0.56||\bar{\kappa}_{3}||_{1}}{(||\kappa_{2}||_{1}-||\text{kl}||_{2}^{2})^{3/2}},
\end{equation}
and assume $\gamma < \frac{1}{2}$. Then 
% Let $f_{1},...,f_{n}$ be symmetric trade-off functions, for all $\alpha \in [\gamma, 1-\gamma]$, we have
\begin{equation}
    G_{\mu}(\alpha+\gamma)-\gamma \leq f_{1}\otimes f_{2}\otimes\cdots \otimes f_{d}(\alpha) \leq G_{\mu}(\alpha-\gamma)+\gamma.
\end{equation}
% where $\mu$ and $\gamma$ are the parameters defined in Theorem 3.4 of \cite{dong2021gaussian}.
\end{lemma}
Given $f_i(\alpha)$ in (\ref{f_of_scalar}), we have
\begin{equation}
    \text{kl}(f) = \frac{c}{A}\log(\frac{A+c}{A-c})
\end{equation}
\begin{equation}
    \kappa_{2}(f) =\log^2(\frac{A+c}{A-c})
\end{equation}
\begin{equation}
    \kappa_{3}(f) = ||\log(\frac{A+c}{A-c})||^3,
\end{equation}
\begin{equation}
    \bar{\kappa}_{3}(f) = \left[\frac{A-c}{2A}\left|1+\frac{c}{A}\right|^3+\frac{A+c}{2A}\left|1-\frac{c}{A}\right|^3\right]\left|\log\left(\frac{A+c}{A-c}\right)\right|^3 .
\end{equation}

Utilizing the composition theorem above, the proposed $n$-bit stochastic quantization can be implemented as $(2^n-1)$ independent stochastic binary mechanisms per coordinate, applied independently across $d$ coordinates. Hence the proposed mechanism is the $(2^n-1)d$-fold composition of the trade-off functions in (\ref{f_of_scalar}) and the corresponding $\mu$ and $\gamma$ are given as follows
\begin{equation}
\mu = \frac{2\sqrt{(2^n-1)d}c}{\sqrt{A^2-c^2}},
\end{equation}

\begin{equation}
\gamma = \frac{0.56\left[\frac{A-c}{2A}\left|1+\frac{c}{A}\right|^3+\frac{A+c}{2A}\left|1-\frac{c}{A}\right|^3\right]}{(1-\frac{c^2}{A^2})^{3/2}\sqrt{(2^n-1)d}},
\end{equation}
which completes the proof.

% In the case of large $d$, which is typical for high-dimensional embeddings in LLMs, $f(\alpha)$ approximates $\mu$.

% Finally, due to the following Lemma on the $n$-fold composition of $\mu$-GDP:

% \begin{lemma}
%     The $n$-fold composition of $\mu_i$-GDP mechanisms is $\sqrt{\mu_1^2 + \dots + \mu_n^2}$-GDP.
% \end{lemma}

% we derive that the stochastic $n$-bit compressor satisfies $\frac{2\sqrt{(2^n-1)d}c}{\sqrt{A^2-c^2}}$-GDP, as it is the $2^n-1$-fold composition of the stochastic binary mechanism, which completes the proof.

\end{proof}

\section{Experimental Details}
\subsection{Metrics}\label{metric_app}
Perplexity is a standard intrinsic metric for language modeling, which measures how well the model predicts a given sequence. 
Formally, for a token sequence $\bm{s}=[s_1,\dots,s_T]$, the perplexity is defined as
\begin{equation}
    \mathrm{PPL}(\bm{s})=\exp\left(-\frac{1}{T}\sum_{t=1}^{T}\log p(s_t|s_{<t})\right),
\end{equation}
where $p(s_t|s_{<t})$ denotes the probability assigned by the model to the next token $s_t$ given the prefix $s_{<t}$. 
Lower perplexity indicates that the model assigns a higher likelihood to the ground-truth sequence, thus reflecting stronger language modeling ability.  

Coherence evaluates the semantic consistency between the target text $Doc$ and the generated output text of LLM $Gen$ by computing the cosine similarity of their embeddings:
\begin{equation}
    COH(Doc,Gen)=\frac{\text{SimCSE}(Doc)\cdot \text{SimCSE}(Gen)}{\|\text{SimCSE}(Doc)\|\cdot \|\text{SimCSE}(Gen)\|},
\end{equation}
where $\text{SimCSE}(\cdot)$ denotes the pretrained sentence embedding model \cite{gao-etal-2021-simcse}.

To evaluate the privacy efficacy of the DP mechanisms, we conduct experiments on both the WikiText-2 and PTB datasets using Llama3 and Qwen2.5. We report the Attack Success Rate for both the embedding inversion attack and the input inference attack. The results are illustrated in Figure \ref{fig:asr_cmp}. As observed, the embedding inversion attack consistently outperforms the input inference attack across all scenarios. This observation aligns with the results reported in InferDPT \cite{tong2025inferdpt}. Consequently, we primarily report the ASR of the embedding inversion attack in the main text.

\begin{figure*}[h]
    \centering
    % Subfigure 1
    % \begin{subfigure}[b]{0.3\textwidth}
    %     \centering
    %     \includegraphics[width=\textwidth]{samples/comparison_plot_llama_ptb.pdf}
    % \end{subfigure}
    % \hfill
    % Subfigure 2

    % \hfill
    % Subfigure 3
    \begin{subfigure}[b]{0.3\textwidth}
        \centering
        \includegraphics[width=\textwidth]{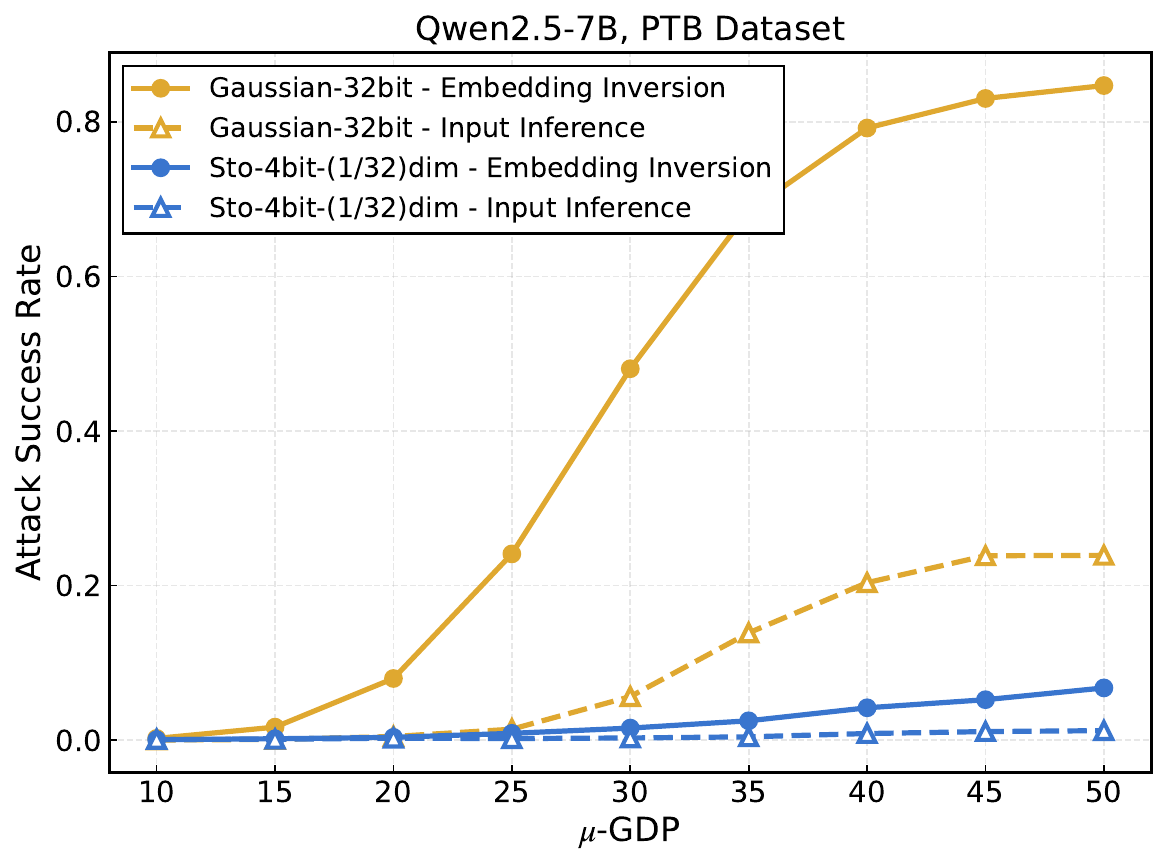}
    \end{subfigure}
    % % \hfill
    % Subfigure 4
    \begin{subfigure}[b]{0.3\textwidth}
        \centering
        \includegraphics[width=\textwidth]{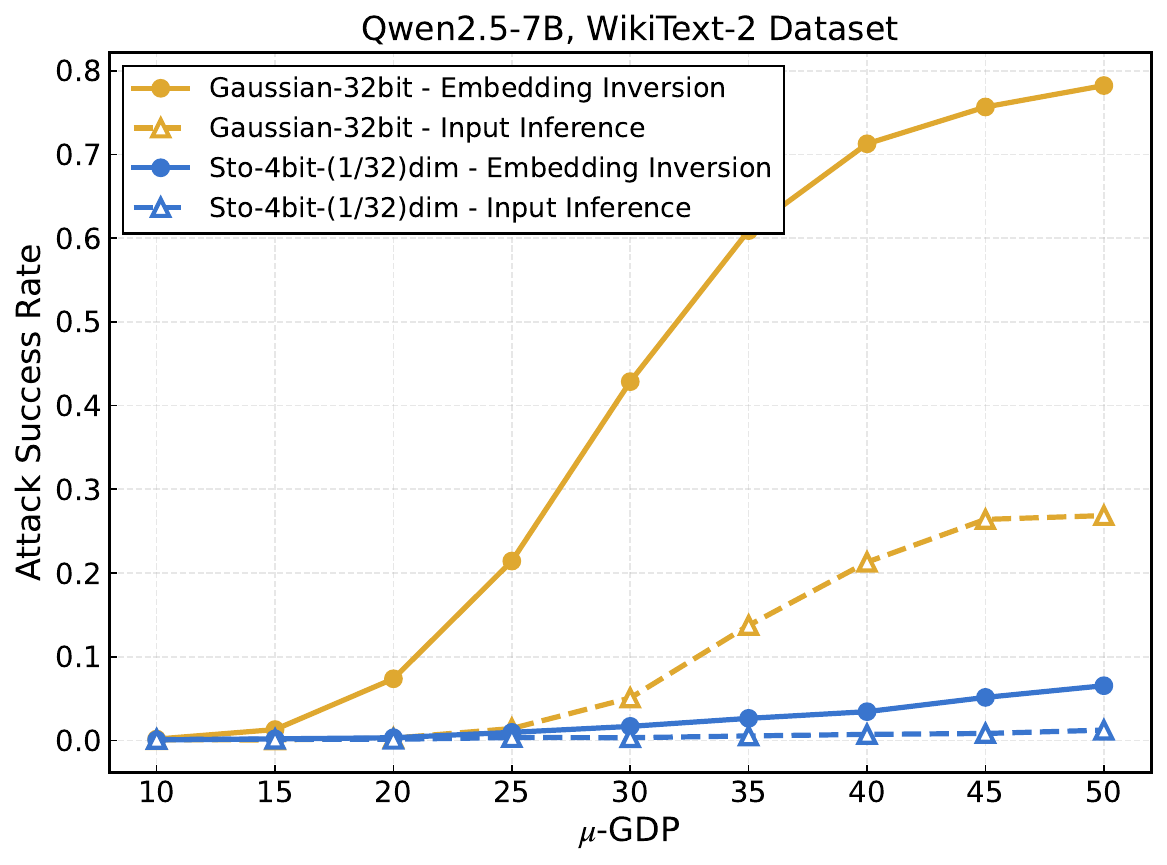}
    \end{subfigure}
        \begin{subfigure}[b]{0.3\textwidth}
        \centering
        \includegraphics[width=\textwidth]{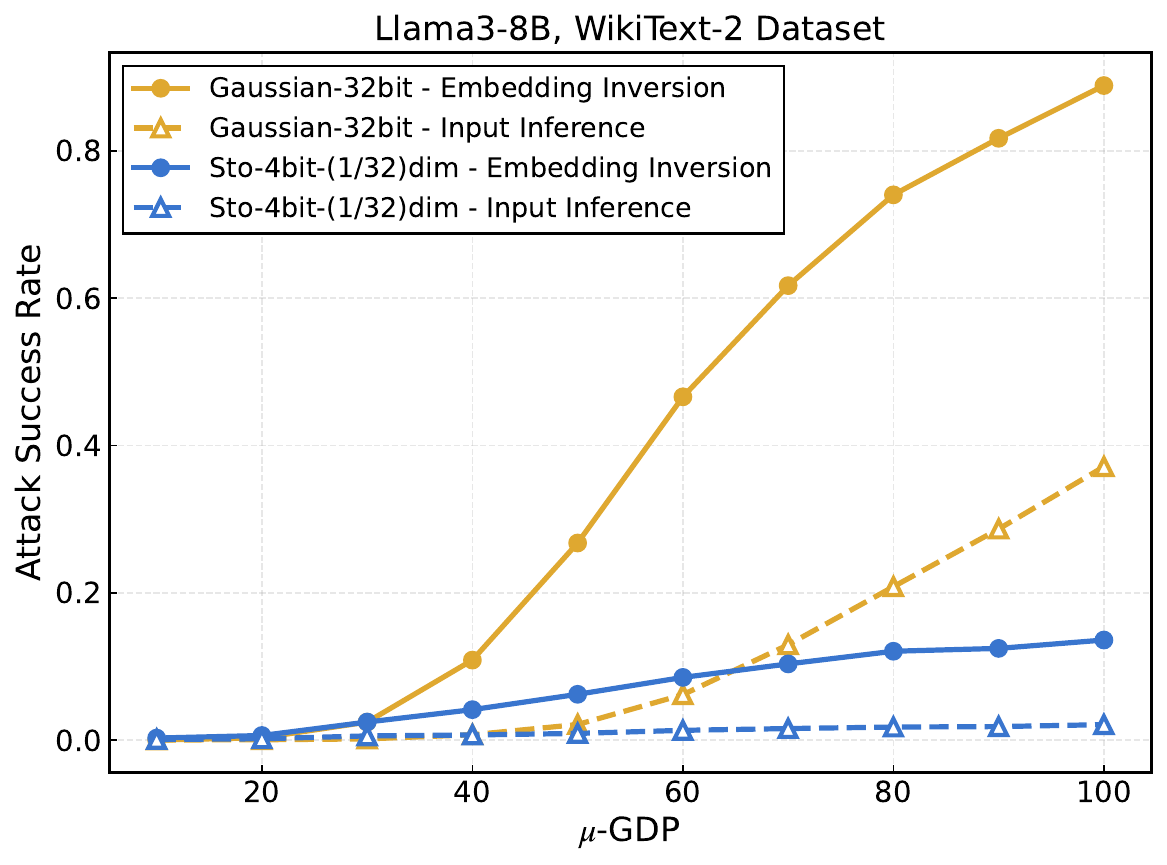}
    \end{subfigure}
    
    \caption{Comparison of Attack Success Rates (ASR) between embedding inversion and input inference attacks across different LLMs, datasets, and DP mechanisms.}
    \label{fig:asr_cmp}
    \vspace{-0.2in}
\end{figure*}

\subsection{Local Instruction Prompt for Extraction in InferDPT}\label{inferspt_prompt}
In InferDPT~\cite{tong2025inferdpt}, the original user input and the noisy output generated by the LLM are fed into a locally deployed lightweight LLM, together with an instruction prompt for local information extraction, as shown below:

\begin{quote}
Your task is to extend the ``Prefix Text''. Use the ``Perturbed Generation'' as your primary writing material for the extension. Extract coherent and consistent text from the ``Perturbed Generation'' and integrate it into your continuation. Ensure seamless alignment with the context established by the ``Prefix Text''. Provide only your ``Extended Text''.\\
--- ``Prefix Text'': \\
--- ``Perturbed Generation'': \\
--- ``Extended Text'':
\end{quote}

\subsection{Details of SnD}\label{snd_detail}
\begin{figure*}[t]
    \centering
    \includegraphics[width=0.8\linewidth]{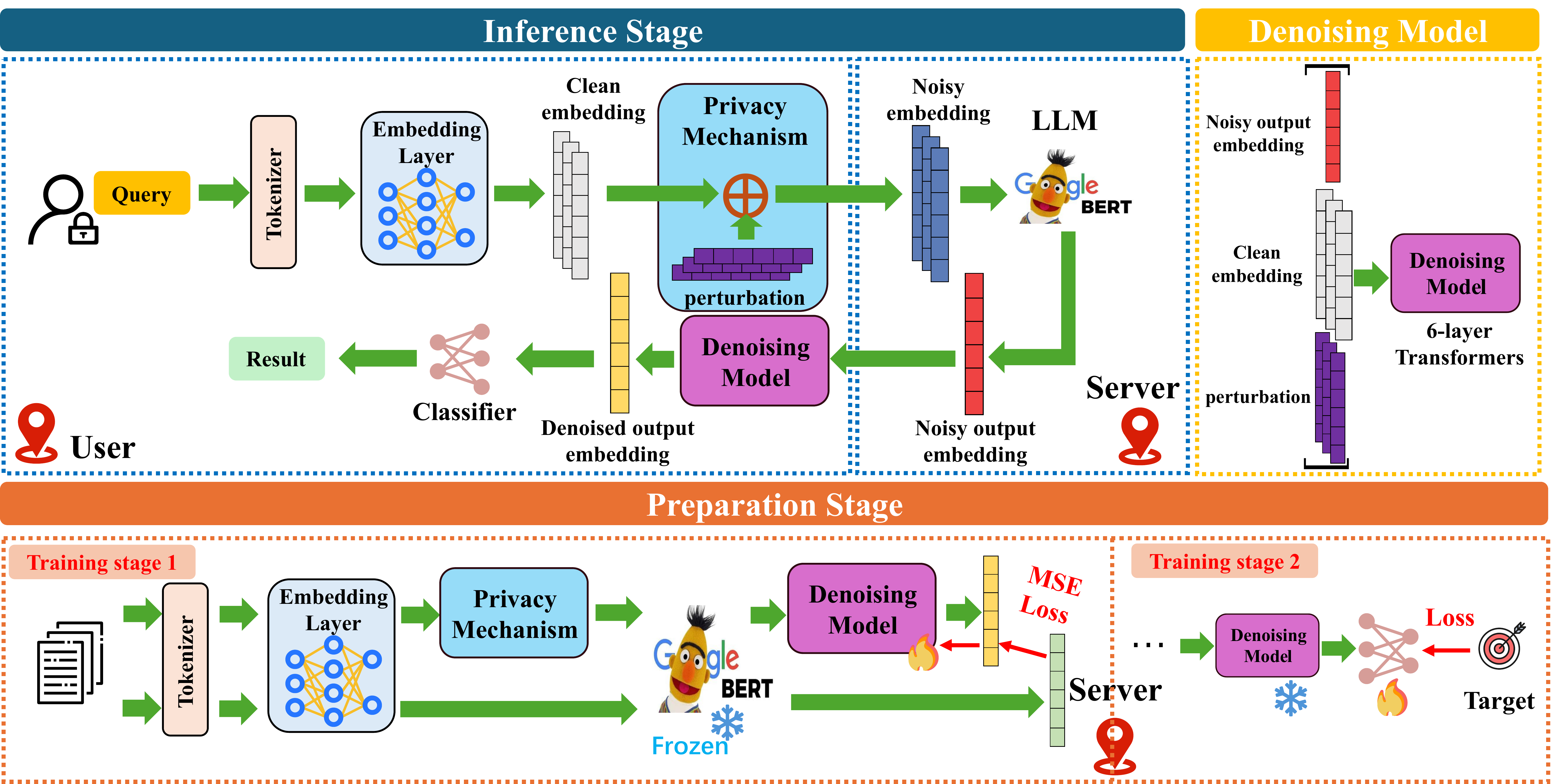}
    % \vspace{-0.2in}
    \caption{An overview of the SnD framework \cite{mai2024split} for NLU tasks.}
    \label{fig:snd}
    \vspace{-0.2in}
\end{figure*}

In SnD~\cite{mai2024split}, as illustrated in Figure~\ref{fig:snd}, the dataset is first converted into token embeddings, which are then perturbed by a local privacy mechanism before transmission.
The resulting noisy embeddings are fed into the LLM (e.g., BERT), which produces noisy output embeddings. These are subsequently denoised by a dedicated denoising model to obtain clean output representations.
For natural language understanding tasks, the denoised output embeddings of each evaluation dataset are used to train a fully connected classifier for downstream classification, where higher accuracy or AUC indicates better utility of the denoised embeddings.
Specifically, as shown in the Training stage 1 of Figure~\ref{fig:snd}, we follow SnD~\cite{mai2024split} to train a denoising model composed of six Transformer layers on the server, which aims to reconstruct the target output embeddings of the LLM given the perturbed input embeddings produced by the privacy mechanism.
The input to the denoising model is the concatenation of the clean token embeddings, the perturbation vector generated by the privacy mechanism, and the noisy output embeddings.
The model parameters are optimized using the Adam optimizer to minimize the mean squared error (MSE) between the denoised and target output embeddings.
% To evaluate the performance on downstream NLU tasks, the denoiser is fixed and shared with the user after training. Afterwards, each evaluation dataset (e.g., MRPC and QQP) is first processed through the privacy mechanism, the LLM server, and the local denoiser to obtain the denoised embeddings. Subsequently, a lightweight fully connected classifier is trained locally on these denoised embeddings to evaluate the accuracy or AUC of the classifier on downstream NLU tasks. 
%%%%%%%%%%%%%%%%%%%%%%%%%%%%%%%%%%%%%%%%%%%%%%%%%%%%%%%%%%%%%%%%%%%%%%%%%%%%%%%
%%%%%%%%%%%%%%%%%%%%%%%%%%%%%%%%%%%%%%%%%%%%%%%%%%%%%%%%%%%%%%%%%%%%%%%%%%%%%%%

\subsection{Additional Experimental Results}
Experimental results on the CNN/Daily Mail dataset are presented in Table \ref{tab:results_cnndm}, demonstrating that the proposed mechanism achieves a superior privacy-utility trade-off compared to InferDPT. Figure \ref{samples1} and \ref{samples2} provide qualitative examples of perturbed inputs and their corresponding outputs. The remote LLM is Qwen2.5-7B, and the local model is the Llama2-7B-4bit as adopted in InferDPT \cite{tong2025inferdpt}. As illustrated, InferDPT often causes the remote LLM to generate unintelligible text, forcing the system to rely heavily on the local model for response generation. However, the perturbations in the remote output occasionally degrade the local model's performance, resulting in garbled or incoherent sequences. In contrast, our proposed method leverages soft prompt to mitigate the utility loss of the remote LLM, enabling it to better capture the semantic essence of the user's input and restore the capabilities. 
% While the generated output may not perfectly match the reference text—partly due to the inherent stochasticity of LLMs even with clean inputs—it remains highly coherent and semantically consistent with the original input. This effectively restores the remote LLM's capabilities and underscores the performance advantages of our approach over InferDPT.

\begin{table*}[!htbp]
\centering
\caption{Comparison with the CNN/Daily Mail dataset between the proposed method and InferDPT under different ASRs. }
\label{tab:results_cnndm}
\scriptsize
\setlength{\tabcolsep}{3pt}
\renewcommand{\arraystretch}{1.15}
\begin{tabular}{l|cccc|c|ccc}
\hline
\multirow{2}{*}{Method} & \multicolumn{4}{c|}{\textbf{InferDPT}} & 
\multirow{2}{*}{\makecell{\textbf{DEL}}} &
\multicolumn{3}{c}{\textbf{DEL + Local Model}} \\
\cline{2-5} \cline{7-9}
 & \makecell{w/o local \\ model} & \makecell{with local \\ DialoGPT} & \makecell{with local \\ OPT-1.3B} & \makecell{with local \\ Llama2-4bit} &  & DialoGPT & OPT-1.3B & Llama2-4bit \\
\hline
\multicolumn{9}{c}{\textbf{Llama3-8B + CNN/Daily Mail}} \\
\hline
ASR = 0.02 & 0.506 & 0.523 & 0.608 & 0.690 & 0.600 & 0.527 & 0.578 & \textbf{0.741} \\
ASR = 0.10 & 0.508 & 0.518 & 0.601 & 0.693 & 0.740 & 0.526 & 0.659 & \textbf{0.760} \\
ASR = 0.15 & 0.513 & 0.522 & 0.605 & 0.698 & \textbf{0.771} & 0.526 & 0.659 & 0.769 \\
ASR = 0.20 & 0.511 & 0.525 & 0.611 & 0.705 & \textbf{0.794} & 0.526 & 0.656 & 0.787 \\
\hline
\multicolumn{9}{c}{\textbf{Qwen2.5-7B + CNN/Daily Mail}} \\
\hline
ASR = 0.02 & 0.497 & 0.518 & 0.598 & 0.712 & 0.651 & 0.525 & 0.607 & \textbf{0.727} \\
ASR = 0.10 & 0.497 & 0.517 & 0.590 & 0.734 & \textbf{0.790} & 0.527 & 0.673 & 0.784 \\
ASR = 0.15 & 0.507 & 0.520 & 0.574 & 0.731 & \textbf{0.817} & 0.526 & 0.667 & 0.804 \\
ASR = 0.20 & 0.504 & 0.519 & 0.597 & 0.730 & \textbf{0.830} & 0.526 & 0.678 & 0.808 \\
\hline
\end{tabular}
% \vspace{-0.1in}
\end{table*}

\begin{figure}[htbp]
    \centering
    % \includegraphics[width=0.95\textwidth]{samples/sample3.pdf} \\
    % \vspace{0.1cm}
    
    \includegraphics[width=0.95\textwidth]{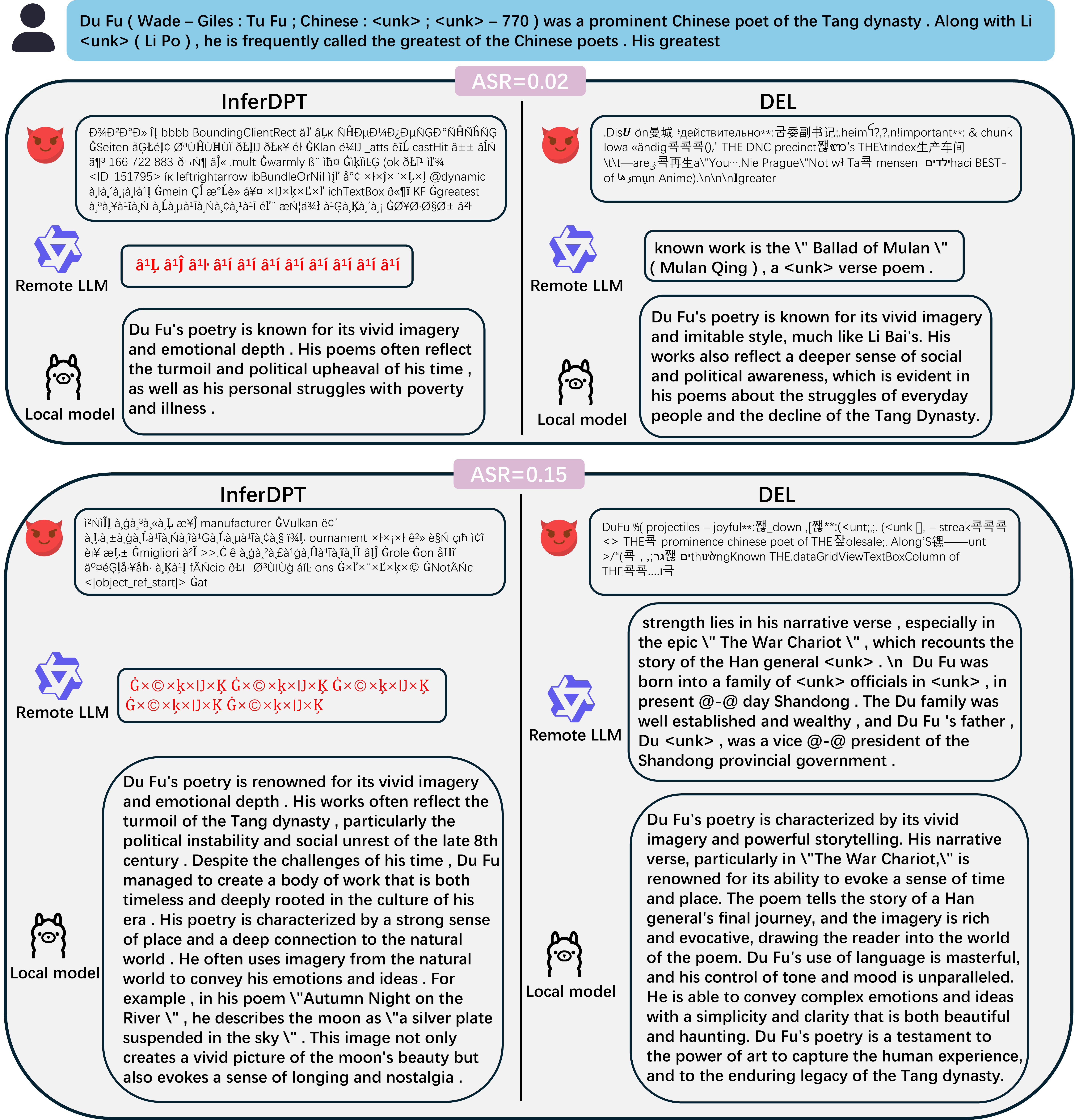}
    % \vspace{-0.1cm}
    \caption{Comparison between the proposed method and InferDPT on perturbed inputs and LLM generation outputs.}
    \label{samples2}
\end{figure}

\begin{figure}[t]
    \centering
    
    % \includegraphics[width=0.95\textwidth]{samples/sample1.pdf} \\
    % \vspace{0.1cm} % 调整间距
    
    \includegraphics[width=0.95\textwidth]{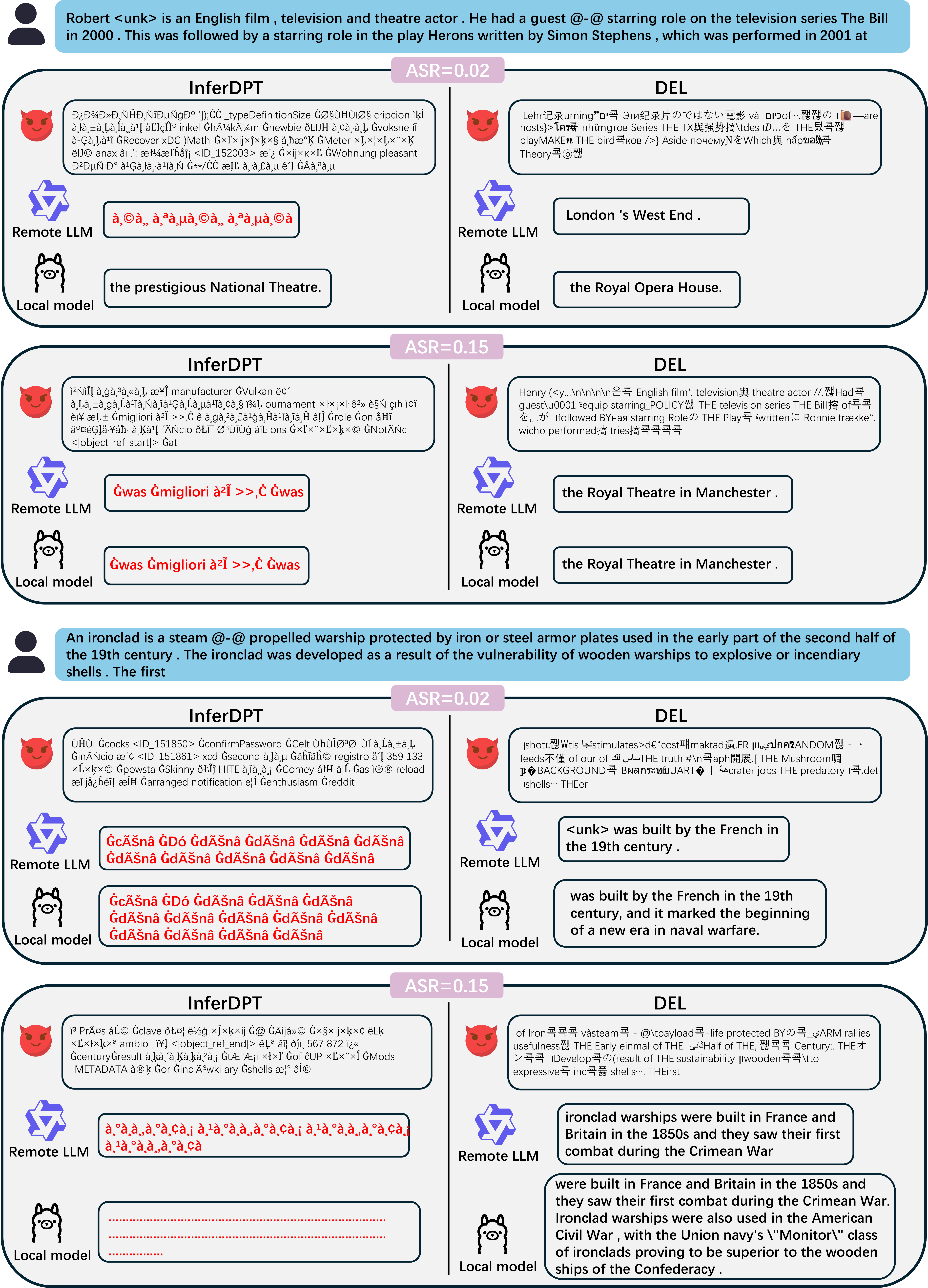} \\
\caption{Comparison between the proposed method and InferDPT on perturbed inputs and LLM generation outputs.}
\label{samples1}
\end{figure}

\end{document}